\documentclass[a4paper]{scrartcl}

\usepackage{multirow}
\usepackage[utf8]{inputenc}
\usepackage[T1]{fontenc}
\usepackage[english]{babel}
\usepackage{amsmath,amssymb,amsfonts,siunitx,commath, amsthm}
\usepackage{tikz,url}
\usepackage{geometry}
\usepackage{mathtools}
\mathtoolsset{showonlyrefs}
\usepackage{subcaption}
\usepackage{algorithm}
\usepackage{algpseudocode}
\usepackage{hyperref}
\usepackage{enumerate}
\usepackage{listings}
\usepackage[normalem]{ulem}
\usepackage{dsfont}
\usepackage{graphicx}
\usepackage{comment}
\DeclareMathOperator*{\argmax}{arg\,max}
\DeclareMathOperator*{\argmin}{arg\,min}

\makeatletter
\newcommand*{\rom}[1]{\expandafter\@slowromancap\romannumeral #1@}
\makeatother

\newcommand{\R}{\mathbb{R}}

\newcommand{\Z}{\mathbb{Z}}
\newcommand{\T}{\mathbb{T}}

\newcommand{\N}{\mathbb{N}}

\newcommand{\wN}{\mathcal{N}_w}
\newcommand{\diag}{\mathrm{diag}}

\newcommand{\n}{\phantom{0}}

\newcommand{\dx}{\,\mathrm{d}}

\newcommand{\uproman}[1]{\uppercase\expandafter{\romannumeral#1}}

\DeclarePairedDelimiter\ceil{\lceil}{\rceil}

\newcommand{\sfract}[2]{{\raisebox{.2em}{$#1$}\left/\raisebox{-.2em}{$#2$}\right.}}

\renewcommand\norm[1]{\left\lVert#1\right\rVert}
\newcommand\independent{\protect\mathpalette{\protect\independenT}{\perp}}
\def\independenT#1#2{\mathrel{\rlap{$#1#2$}\mkern2mu{#1#2}}}

\theoremstyle{plain}
\newtheorem{lemma}{Lemma}[section]
\newtheorem{theorem}[lemma]{Theorem}

\newtheorem{proposition}[lemma]{Proposition}
\newtheorem{remark}[lemma]{Remark}
\theoremstyle{definition}
\newtheorem{definition}[lemma]{Definition}
\newtheorem{example}[lemma]{Example}

\begin{document}
\title{Learning Sparse Mixture Models}

\author{
Fatima Antarou Ba\footnotemark[1]
}

\maketitle

\date{\today}

\footnotetext[1]{
TU Berlin,
Stra{\ss}e des 17. Juni 136, 
D-10587 Berlin, Germany,
\{fatimaba\}@math.tu-berlin.de.
} 
\begin{abstract}
This work approximates high-dimensional density functions with an ANOVA-like sparse structure by the mixture of wrapped Gaussian and von Mises distributions. When the dimension $d$ is very large, it is complex and impossible to train the model parameters by the usually known learning algorithms due to the curse of dimensionality. Therefore, assuming that each component of the model depends on an a priori unknown much smaller number of variables than the space dimension $d,$ we first define an algorithm that determines the mixture model's set of active variables by the  Kolmogorov-Smirnov and correlation test. Then restricting the learning procedure to the set of active variables, we iteratively determine the set of variable interactions of the marginal density function and simultaneously learn the parameters by the Kolmogorov and correlation coefficient statistic test and the proximal Expectation-Maximization algorithm. The learning procedure considerably reduces the algorithm's complexity for the input dimension $d$ and increases the model's accuracy for the given samples, as the numerical examples show.
\end{abstract}
\tableofcontents
%-------------------------------------------
\section{Introduction}
One of the most recurrent problem of multivariate function approximation theory problems is the curse of dimensionality. An algorithm is said to face the curse of dimensionality if the algorithm depends exponentially on the dimension of the data. In order to circumvent or solve the problem, several authors have focused on the study of sparse functions with respect to their arguments. A widespread example appears in compressive sensing~\cite{Tran2020MultilinearCL}, where the target function can be spanned precisely by assuming that the input vector is sparse with respect to the $\norm{\cdot}_0$-norm. Another theory assumes that the target function can be decomposed into a sum or product of much smaller dimensional functions~\cite{Pereyra1973, BeylkinMohlenkamp2005,  Grelier2019LearningHP}. In the specific case of ANOVA decomposition~\cite{Owen2019}, it is assumed that only a minimal number of ANOVA terms whose dimension is minimal compared to $d$ are relevant. This implies that the function to be approximated can be factorized into a sum of functions that depend on only a limited number of variables~\cite{Potts2019, MARA2012115}, i.e., only a certain number of variables interact with each other. Thus the notion of superposition dimension and truncation dimension~\cite{Tran2020MultilinearCL} was introduced to penalize the number of ANOVA terms (equal to $2^d$) and, in addition, the dimension of each of them. Several fruitful pieces of research have been done in this sense, as in regression problems \cite{ANOVApotts2021, MARA2012115} and density function approximations \cite{hertrich2021sparse, HertrichPCA2021, BEYLKIN201994}. Therefore, we introduced finite sparse mixtures models which are inspired by the ANOVA decomposition of sparse functions. Indeed we assume that each mixture component may only depend on a smaller of variables interaction than the space dimension $d$.
\subsection{Prior work}
Given $\set{x^n}_n$ data of a multivariate random variable $X\in \T^d$ of potentially very large dimension, the objective of our work is to approximate the density function $f$ through a mixture of wrapped Gaussian or von Mises distribution models. One of the best known methods is Expectation-Maximization which maximizes the likelihood of the data. 
It should be noted that in the case where the dimension is high it is impossible to apply the algorithm naively without prior knowledge of the sparsity of the density function $f$. Thus, in a previous paper~\cite{hertrich2021sparse} we tried to take into account the sparsity assumption of the density function $f$ of the mixture model. The algorithm proved to be very effective in approximating periodic B-splines, the first Friedman function and in image classification. 
\subsection{Our contribution}
This current paper is an extension of our previous work "Sparse ANOVA Inspired Mixture Models" \cite{hertrich2021sparse}. In particular we deal with improvement of learning algorithm by first determining the active variables of the density function. Then we restricted the study to the set of active variables. This approach is even more efficient if we assume that some variables do not play any role in the approximation of $f$. This considerably reduces the computational time and space. Thus we will assign masses to the variables according to the amount of information they contain. Thus it is possible to obtain an accurate approximate the density function by its marginal which contains the variables with the most information.

\subsection{Outline of paper}

Section\ref{sec:notation} introduced the notation. In section~\ref{sec:SPAMM}, we have introduced a sparse mixture model from the parametric family of multivariate wrapped Gaussian and von Mises distribution. Furthermore, we have derived the marginal and the conditional density function of the wrapped Gaussian Distribution, which will later help us approximate the target density function iteratively.

In section~\ref{sec:learning}, we have implemented an algorithm that determines the set of active variables of a sparse mixture model by the Kolmogorov-Smirnov and correlation coefficient 
test. Therefore the model learning can be restricted to active variables set $\mathcal{A}$, which will considerably reduce the complexity of the model training if we assume that $\abs{\mathcal{A}} \ll d$. 

In section~\ref{sec:learning}, we will define an Algorithm that will iteratively estimate the set of interacting variables and the parameters of the marginal density function as well. Later in section~\ref{sec:numerics}, we will test our model on sparse mixtures of wrapped Gaussian, B-splines function, and the California Housing prices data. 
\section{Preliminaries and notation}\label{sec:notation}

%-------------------------------------------

\section{Sparse Mixture Models} \label{sec:SPAMM}

\subsection{Sparse additive  Model}\label{sec:general}

%\textcolor{blue}{model and marginals
%	\begin{equation}\label{eq:sparse_density_function}
%		f(x| \alpha, \theta)= \sum_{k=1}^K \alpha_{u_k} p(x_{u_k} | \theta_{u_k})
%\end{equation}}

Under similar assumption as \cite[section~2]{hertrich2021sparse} , we try to approximate the density function $g$ of an unknown distribution given a finite number of weighted samples $\mathcal{X}=\set{\left(x^n, w_n\right)}_{n \in [N]}$ by a finite dimensional sparse mixture model, whose probability density function (pdf) is given by  
\begin{equation}\label{eq:sparse_density_function}
	f(x\mid \alpha, \theta)= \sum_{k=1}^K \alpha_{u_k} p(x_{u_k} \mid \theta_{u_k})
\end{equation} 
where $u_k \in U \subset \mathcal P([d]), \alpha=(\alpha_{u_k})_{k=1,\ldots, K} \in \Delta_{K}, \theta=(\theta_{u_k})_{u_k\in U}$ and $p(\cdot \mid \theta_{u_k})$ is a probability density function with $\abs{u_k}$-dimensional parameter $\theta_{u_k}$. We will here consider samples $(x_n)_{n\in N}$ which are equally weighted, i.e $w_n=1$ for all $n\in N$.
Similarly to \cite{hertrich2021sparse}  we will also assume that the index set $u_k$ may not be pairwise different, i.e there may exist $k,t \in [K],$ such that $u_t =u_k$ but $t\neq k$. Thus denotes by $K_u$ the number of mixture components $p(\cdot \mid \theta_{u_k})$ such that $u=u_k, k=1,\ldots, K.$ Mixture models, whose density function has the form \eqref{eq:sparse_density_function} are called \textit{sparse mixture model} (\textit{sparse MM}).
%We will further assume that for all $i \in [d]$ there exits at most two elements $u, \tilde{u} \in U$ such that $u\cap [i]= \tilde{u}\cap [i].$ \textcolor{blue}{(The approximation error of non wrapped density function with one component become larger. Testing a new method with can return an optimal number of components given an $u \in U$. Thus we can assume that the $u \in U$ can appear twice in the density function and the approximation error should also be much smaller.)}
The parametric family of sparse wrapped Gaussian distribution with both diagonal and full covariance matrix on one side and the family of sparse von Mises distribution on the other will be used to approximate the unknown target density function. 
Recall that the wrapped Gaussian distribution is obtained by wrapping the Gaussian distribution around the torus. Indeed if $Y$ is a Gaussian distributed random variable (RV), the corresponding wrapped Gaussian RV $X$ is given by $X=Y\text{mod } T=Y-T\ceil{\frac{Y}{T}},$ where $T>0$ denotes the period. Since we are interested in approximating $1$-periodic functions on the unit torus, then the wrapped random variable becomes $X=Y-\ceil{Y}.$ There pdf are defined as
\begin{equation}\label{eq:defwrappedpdffull}
	p_G\left(x_{u_k}\mid \theta_{u_k}= (\mu_{u_k}, \Sigma_{u_k})\right) = \sum_{l \in \Z^{\abs{u_k}}} \mathcal{N}(x_{u_k}+l \mid \mu_{u_k}, \Sigma_{u_k}) = \wN(x_{u_k}\mid \mu_{u_k}, \Sigma_{u_k}),
\end{equation}   
where ($\wN$) $\mathcal{N}$ denotes the pdf of the $\abs{u_k}$-dimensional (wrapped) Gaussian distribution with mean $\mu_{u_k} \in \T^{\abs{u_k}}$ and symmetric positive definite (SPD) covariance matrix $\Sigma_{u_k} \coloneqq \Sigma_{u_ku_k} \in \R^{\abs{u_k} \times \abs{u_k}}.$
If the wrapped Gaussian distribution has a diagonal covariance matrix $\Sigma = \diag(\Sigma_1, \ldots, \Sigma_{\abs{u_k}})$ then its  pdf is simplified to a product of univariate wrapped Gaussian density function as
\begin{equation}\label{eq:defwrappedpdfdiag}
	p_G\left(x_{u_k}\mid \theta_{u_k}= (\mu_{u_k}, \Sigma_{u_k})\right) = \sum_{l \in \Z^{\abs{u_k}}} \prod_{i \in u_k} \mathcal{N}(x_{i}+l_i \mid \mu_{i}, \Sigma_{i}) = \prod_{i \in u_k}\wN(x_{i}\mid \mu_{i}, \Sigma_{i}),
\end{equation}
where ($\wN$) $\mathcal{N}$ is the pdf of the univariate (wrapped) Gaussian density function with parameters $\mu_i$ and $\sigma^2=\Sigma_{i}.$ Since it is practically impossible to numerically compute the probability function of the wrapped Gaussian distribution, and due to the assumption on its covariance matrix, which is positively definite, it can been shown that 
\begin{equation*}
	\wN(x_{u_k}\mid \mu_{u_k}, \Sigma_{u_k}) \approx \wN^B(x_{u_k}\mid \mu_{u_k}, \Sigma_{u_k}),
\end{equation*}
where
\begin{equation*}
	\wN^B(x_{u_k}\mid \mu_{u_k}, \Sigma_{u_k}) \coloneqq \sum_{l \in ([-B, B]\cap \Z)^{\abs{u_k}}}\mathcal{N}(x_{u_k}+l \mid \mu_{u_k}, \Sigma_{u_k})
\end{equation*}
for a suitably chosen $B \in \N.$ For instance \cite{Jona_Lasinio_2012, MardiaJupp2000} has derived some values of $B$ depending on the standard deviation $\sigma$ for $T=2\pi, d=1,$  where $\wN^B(x_{u_k}\mid \mu_{u_k})$ approximates gut the ground truth density function $\wN(x_{u_k}\mid \mu_{u_k})$.  It has been showed by \cite{MardiaJupp2000} for
\begin{equation}
	B= \begin{cases}
		1, & \text{if } \sigma \geq 2\pi,\\
		0, &\text{ otherwise},
	\end{cases}
\end{equation}  
and by \cite{Jona_Lasinio_2012} for
 \begin{equation*}
B= \begin{cases}
 	1, & \text{if } \sigma<2\pi/3,\\
 	2, & \text{if }  2\pi/3 \leq \sigma < 4\pi/3,
 \end{cases}
 \end{equation*}
the approximation is very accurate. As the space dimension increases, then $B$ also increased.  
 Thus we will consider the truncated function $\wN^B(x_{u_k}\mid \mu_{u_k}, \Sigma_{u_k})$ instead of $\wN(x_{u_k}\mid \mu_{u_k}, \Sigma_{u_k})$ in the rest of the paper. 
 \begin{remark}
 The pdf of the wrapped Gaussian distribution from \eqref{eq:defwrappedpdffull} can be interpreted as the marginal density function with respect to $(X_u, L_u) \in (\T^{\abs{u_k}}, \Z^{\abs{u_k}})$ of the joint pdf
 \begin{equation*}
 	f(x,l\mid \mu, \Sigma)=\mathcal{N}(x+l\mid \mu,\Sigma),
 \end{equation*}
where $\mu$ and $\Sigma$ are the wrapped normal distribution parameters of $X,$ and the hidden variable $L_u$ denotes the number of winding, i.e $Y=X+L \sim \mathcal{N}(\cdot\mid \mu+L, \Sigma)$ 
 \end{remark}
 The von Mises distribution, which represents the restriction of the pdf of an isotropic normal distribution to the unit circle  has the pdf
\begin{equation*}
	p_M(x_{u_k} \mid \theta_{u_k} = (\mu_{u_k}, \kappa_{u_k})) = \prod_{i \in u_k}\frac{1}{I_0(\kappa_i)}\exp\left(\kappa_i\cos\left(2\pi(x_i-\mu_i)\right)\right),
\end{equation*}
where $\mu_{u_k} \in \T^{\abs{u_k}}$ represents mean and $\kappa_{u_k} \in \R^{\abs{u_k}}_+$ and $I_0$ is the modified Bessel function of the first kind of order $0$.

To ensure a good approximation accuracy by the parametric family of sparse mixture models of wrapped and von Mises distribution, we assume furthermore that the ground function $g$ is smooth enough and has a compact support, since  Gaussian Mixture Models (GMM) has proved to be good approximators for continuous density functions with compact support \cite{ATHANASSIA2010}.

\begin{comment}
\begin{theorem}\label{thm:density_func_appprox}
Let $W \subset P([d])$ and $g: \T^d \rightarrow \R_+$ be a \textcolor{blue}{periodic?} sparse density function of the form
\begin{equation}
g = \sum_{w \in W} g_w,
\end{equation}
where $g_w$ are weighted $\abs{w}$-dimensional density functions (or a linear combination). \textcolor{blue}{ToDo: State some smoothness assumption on $g$.} Then there exists a sparse mixture model of wrapped Gaussian distribution and von Mises distribution, with density function $f$ such that
\begin{equation*}
\norm{g- f}_{\infty} \leq \varepsilon ,
\end{equation*}
where $\varepsilon>0$ and $f$ is of the form~\ref{eq:sparse_density_function}.
\end{theorem}

\begin{proof}
content...
\end{proof}

\end{comment}
Under the same assumption as above we will derive a form of the marginal density function of $f$, where $f$ is defined as \eqref{eq:sparse_density_function}. Indeed we will introduce later in section~\ref{sec:learning} an algorithm that iteratively approximate the marginal density function of $f$. 
\begin{definition}\label{def:marginal_density_function}
	Let $X \in \T^d$ be a continuous random variable with probability density function $f$. For every $u \subset [d],$ the marginal probability density function with respect to $X_u \coloneqq (X_j)_{j \in u}$ is defined as 
	\begin{equation*}
		f_{X_u}(x_u\mid \theta_u) = P_{u}f(x \mid \theta)= \int_{\T^{d - \abs{u}}} f(x_u, x_{u^c}\mid \theta)dx_{u^c},
	\end{equation*}
	where $u^c = [d]\setminus u$ and $P_u: \T^d \rightarrow \T^{\abs{v}}$ is the \textit{projection operator}.
\end{definition}
The linearity of the integral immediately implies that the marginal distribution of a mixture model is equal to the mixture of the marginal of each mixture component and thus the linearity of the projection operator. Since two different components may have the same marginal (i.e the same parameters), then they are put together by summing their mixing weights. This implies, that number of components of the marginal is smaller or equal to the number of components of the ground mixture model. Furthermore if the multivariate random variable $X$ is componentwise independent or (wrapped) Gaussian distributed with parameter $\theta$ then the marginal distribution with respect to the subset of random variable $X_u, u \subset [d]$ is of the same family as the ground distribution with parameters $\theta_u = (\theta_j)_{j \in u}.$ For the special case of wrapped Gaussian distribution with dependent random variables the assumption also holds. Before stating the theorem on the marginals of sparse mixture models of wrapped Gaussian or von Mises distributions, let us recall first the marginal and conditional distribution of a multivariate wrapped Gaussian distribution.

\begin{lemma}\label{lem:conditional_distr_marginal_dist}
	Let $u \subset [d], n \coloneqq \abs{u}$ and $n_c = d-n$. Let furthermore $X =\left(X_u, X_{u^c}\right) \in \left(\T^{n}, \T^{n_c}\right)$ be a multivariate continuous random variable of a wrapped Gaussian distribution, i.e
	\begin{equation*}
		X \sim \wN\left(\mu, \Sigma\right),
	\end{equation*} 
	with parameters
	\begin{equation*}
		\mu \coloneqq \begin{pmatrix}
			\mu_{u}\\
			\mu_{u^c}
		\end{pmatrix}, \Sigma \coloneqq \begin{pmatrix}
			\Sigma_{uu} & \Sigma_{uu^c} \\ 
			\Sigma_{u^cu} & \Sigma_{u^cu^c}
		\end{pmatrix},
	\end{equation*}
	such that $\mu_{u} \in \T^n, \mu_{u^c} \in \T^{n_c}$ are the mean  and $ \Sigma_{uu} \in  \R^{n\times n},\Sigma_{u^cu^c} \in  \R^{n_c\times n_c}, \Sigma_{uu^c} \in \R^{n\times n_c}, \Sigma_{u^cu} \in \R^{n_c\times n}$ are positive definite covariance matrices parameters. Then the marginal distribution of $X_u$ and $X_{u^c}$ are also a wrapped Gaussian distribution, such that
	\begin{equation*}
		X_u \sim \wN\left(\mu_u, \Sigma_{uu}\right), X_{u^c} \sim \wN\left(\mu_{u^c}, \Sigma_{u^cu^c}\right).
	\end{equation*}
	The conditional distribution of $X_{u}\mid X_{u^c}$ is given by
	\begin{equation}
		X_{u^c}\mid (X_u,L_u)=(x_u,l_u) \sim \wN\left(\bar{\mu}, \bar{\Sigma}\right),
	\end{equation}
	with
	\begin{align*}
		\bar{\mu} & = \mu_{u^c} + \Sigma_{u^cu}\left(\Sigma_{uu}\right)^{-1}\left(x_u+l_u-\mu_{u}\right),\\
		\bar{\Sigma} &= \Sigma_{u^cu^c} - \Sigma_{u^cu}\left(\Sigma_{uu}\right)^{-1}\Sigma_{uu^c},
	\end{align*}
where $L_u$ denotes the $\abs{u}$-dimensional winding number~\cite{Jona_Lasinio_2012}
\end{lemma}

\begin{theorem} \label{thm: marginal_mixture_model}
	Let $X \in \T^d$ be a continuous random variable of a sparse mixture model of wrapped Gaussian distribution, with density function $f: \T^d \rightarrow \R_+$. Let furthermore $X_u:=(X_j)_{j \in u}, u \subset [d]$ be an $\abs{u}$-dimensional random variable and $\xi_k \coloneqq u\cap u_k$. Then the marginal distribution with respect to $X_u$ is also a sparse mixture model with density function
	\begin{equation}
		f_{X_u} = \sum_{t=1}^{T_u} \alpha_t p(\cdot \mid \theta_{u_t}),
	\end{equation}
	where $T_u \leq K$ and $u_t$ is element of
	\begin{equation*}
		U_{X_u} = \set{v\cap u \mid v \in U }.
	\end{equation*} 
and $U$ is the collection of the indices of interacting variables of $f.$
	The mixing weights and density functions of the marginal distribution are respectively
	\begin{equation*}	
		\alpha_t = \sum_{k=1}^K \alpha_{k} \chi_{\set{\theta_{\xi_k}=\theta_{u_t}}},
		\quad p(\cdot \mid \theta_{u_t}) = \begin{cases}
			1, & \text{if } u_t = \emptyset, \\
			P_{\xi_k}p(\cdot \mid \theta_{u_k}), &  \text{ otherwise},
		\end{cases} 
	\end{equation*}
	such that $\sum_{t}^{T_u}\alpha_t = 1$ and for $u_t \neq \emptyset$
	\begin{equation*}
		\theta_{u_t} = \left(\mu_{u_t}, \Sigma_{u_tu_t}\right) = \left(\mu_{\xi_k}, \Sigma_{\xi_k\xi_k}\right).
	\end{equation*}
\end{theorem}
\begin{proof}
	By definition of the marginal density function and by the linearity of the integral, the marginal density of the mixture model holds
	\begin{align*}
		f_{X_u} &= \int_{\T^{d - \abs{u}}} f dx_{u^c}= \sum_{k=1}^{K} \alpha_k \int_{\T^{d - \abs{u}}} p(\cdot \mid \theta_{u_k})dx_{u^c},\\
		&= \sum_{k=1}^{K} \alpha_k  P_vp(\cdot \mid \theta_{u_k}) = \sum_{k=1}^{K} \alpha_k  P_{\xi_k}p(\cdot \mid \theta_{u_k}).
	\end{align*}
	The definition of the marginal density function implies that
	\begin{equation*}
		p(\cdot \mid \theta_{u_t}) = \begin{cases}
			1, & \text{if } \xi_k = \emptyset, \\
			P_{\xi_k}p(\cdot \mid \theta_{u_k}), &  \text{ otherwise}.
		\end{cases} 
	\end{equation*}
	The theorem on conditional wrapped Gaussian distribution, implies that for each $k\in [K]$ the marginal distribution of each mixture component $k$ is a wrapped Gaussian distribution with parameter
	\begin{equation*}
		\theta_{{u_k}\cap u} = \left(\mu_{\xi_k}, \Sigma_{\xi_k\xi_k}\right),
	\end{equation*}
	if $\xi_k \neq \emptyset.$
	Thus the marginal of the mixture model yields
	\begin{equation*}
		f_{X_u} = \sum_{k=1}^K \alpha_k p(\cdot \mid \theta_{\xi_k}),
	\end{equation*}
	where $p(\cdot \mid \theta_{\xi_k})$ denotes the probability function of the wrapped Gaussian with parameter $\theta_{\xi_k}.$
	If there exists $k\neq t, k,t\in [k],$ such that  $\theta_{\xi_k} =\theta_{u_t\cap u}$ then combine both components by summing up their weights and reduce the number of mixture component to one.
\end{proof}

Theorem~\ref{thm: marginal_mixture_model} shows, that the marginal of a sparse mixture model of a parametric family of wrapped Gaussian or von Mises distribution may contain the uniform distribution as mixing component.
\subsection{Determination of Active Variables}\label{sec:Active_variables}

%\textcolor{blue}{KS applied to $x^n$\\
%	numerical example}

Assuming that the above assumptions are fulfilled, we can considerably reduce the complexity of learning the parameters of the sparse mixture models by removing the independent uniform distributed random variables. Indeed a random variables $X=\left(X_{\mathcal{A}}, X_{\mathcal{A}^c}\right) \in \T^d$ such that $\mathcal{A}\subset [d],$  $X_{\mathcal{A}}$ and $X_{\mathcal{A}^c}$ are independent yields the Bayes theorem
	\begin{equation}
		h(x)=p_1(x_{\mathcal{A}}\mid \theta_1)p_2(x_{\mathcal{A}^c}\mid x_{\mathcal{A}}, \theta_2)
	\end{equation}
where $p_1$ denotes the marginal density function with respect to $X_{\mathcal{A}}$ and $p_2$ the conditional density function of $X_{\mathcal{A}}\mid X_{\mathcal{A}^c}.$ By assumption  $p_2(x_{\mathcal{A}^c}\mid x_{\mathcal{A}}, \theta_2) = p_2(x_{\mathcal{A}^c}\mid \theta_2)$ since both random variables are independent. If we further assume that $p_1$ is a sparse density function having the form \eqref{eq:sparse_density_function} and $p_2$ is the uniform density function then 
\begin{equation}
	h(x)=\sum_{k=1}^{K}\alpha_k p(x_{u_k}\mid \theta_k),
\end{equation}
since the multivariate uniform density function on $\T$ is equal to $1$ everywhere.
Therefore  we can introduce the notion of \textit{active} and \textit{inactive variables} for sparse density functions.
\begin{definition}
Let $X\in \T^d$ be a multivariate random variable with density function 
		\begin{equation}
		f(x)=\sum_{k=1}^{K}\alpha_k p(x_{u_k}\mid \theta_k).
	\end{equation}
a mixture model of wrapped or von Mises mixture model.  The set $\mathcal{A}$ of \textit{active variables} of $f$ by
\begin{equation}
	\mathcal{A}_f :=\set{i \in [d] \mid \exists u \in U: i \in u}.
\end{equation} 
and any random variable $X_i$ such that $i\in \mathcal{A}$ is called \textit{active}.   
\end{definition}
Thus an active random variable $X_i$ is either non uniformly distributed or dependant to some $X_j$ such that $j\in \mathcal{A}.$ Otherwise the random variable is called inactive.
%A  random variable $X_i, i \in [d]$ is called \textit{inactive} if $X_i \sim \mathcal{U}(\T)$ and is independent of $X_j$, i.e $X_i \independent X_j$ for all $j \in [d]\setminus \set{i}.$ This implies that $i \notin u$ for all $u \in U.$ Otherwise it is \textit{active} if there exists at least one $u\in U,$ such that $i\in u$. 
%Note that $\mathcal{A}_f$ is generally not equal to $U.$
Taking as example the density function of the sparse mixture model defined in \eqref{eq:sparse_density_function}, the active set of each mixture component is given by $A_f^{u_k}:=\set{i \in u_k}$, which yields
\begin{equation*}
	\mathcal{A}_f= \bigcup\limits_{k=1, \ldots, K} A_f^{u_k}.
\end{equation*}
Based on $\mathcal{X}$ we can determine iteratively $\mathcal{A}_f$ by checking which features variables are non uniform distributed with the help of Kolmogorov-Smirnov test or which depends to the non-uniform random variables. Since the independence is generally not trivial, we will only test the random variables by correlations. 
% For those variables which fulfill the test, determine additionally if they are also uncorrelated to the other features. The variables that do not fulfill at least one of these tests are the active variables. 
We can explicitly determine the active set of density function $f$ given a large enough number of weighted samples by Algorithm~\ref{alg:det_active_set}  
\begin{algorithm}[!ht]
	\caption{Determine the active set}\label{alg:det_active_set}
	\begin{algorithmic}[1]
		\State \textbf{Input}: $d$-dimensional weighted samples $\mathcal{X}=\set{\left(x^n, w_n\right)_{n\in [N]}}$ with probability density function $f$, $\varepsilon_{KS}, \varepsilon_{c} >0.$ 
		\State \textbf{Output}: the active set $\mathcal{A}_f.$
		\State Determine the index set $\mathcal{S}$ of uniform distributed variables by:
		\For{$i=1, \cdots, d$}
		\If{$D((x^n_i, w_n)_{n\in [N]}) \leq \varepsilon_{KS}$}
			\State Add $\mathcal{S} \leftarrow \set{i} $
		\EndIf
		\EndFor
		\State Set $\mathcal{A}_f \coloneqq [d]\setminus \mathcal{S}$
		\For {$i$ in $\mathcal{S}$}
		\State Compute correlation coefficient vector $C(X_i, X_{\mathcal{A}})$between $X_i$ and $X_{\mathcal{A}}$
		\If{$C(X_i,X_j)\geq \varepsilon_{c}$ for some $j \in \mathcal{A}_f$}
			\State Add $\mathcal{A}_f \leftarrow \set{i}$
		%	\If{$\text{Cor}(X_i,X_t)\geq \varepsilon_{c}$ for some $l \in \mathcal{S}[:i] $}
		%	\State Add $\mathcal{A}_f \leftarrow \set{l}$
		%	\EndIf
		\EndIf
		\EndFor
	\end{algorithmic}
\end{algorithm}
To better understand the concept let us consider two density functions of mixture of wrapped Gaussian distribution which will study in detail along the paper.
	\begin{example}\label{ex:theor_ex_1}
	Consider a  $15$-dimensional density functions
	\begin{equation*}
		f^j(x) 
		\coloneqq
		\sum_{k=1}^2  \alpha_k p_{u_k} (x_{u_k}|\mu_k,\Sigma_k), 
		\quad
		p (x_{u_k}|\mu_k,\Sigma_k) \coloneqq \sum_{l\in\Z^{|u_k|}}\mathcal N(x_{u_k}+l|\mu_k,\Sigma_k),
	\end{equation*} 
	where $j=1,2.$
	The first function $f^1$ has the parameters
	\begin{align*}
		U&\coloneqq\set{u_1,u_2 }
		= \set{\{0,1,8\},\{0,14\}},\\ 
		\alpha
		&\coloneqq(0.7,0.3),\\
		\mu 
		&\coloneqq \frac12 \left( (1,1,1)^{\intercal},(.5,.5)^{\intercal} \right),\\
		\Sigma_k &\coloneqq \sigma^2 I_{\abs{u_k}}, \sigma^2 \coloneqq 0.001.
	\end{align*}
	The second function $f^2$ has the parameters 
	\begin{align*}
		U&\coloneqq\set{u_1,u_2, u_3 } = \set{\{0,1,8\},\{0,14\}, \{5\}},\\ 
		\alpha
		&\coloneqq(0.7,0.2, 0.1),\\
		\mu 
		&\coloneqq \frac12 \left( (1,1,1)^{\intercal},(0.5,0.5)^{\intercal}, 0.3 \right),\\
		\Sigma_k &\coloneqq \sigma^2 I_{\abs{u_k}}, \sigma^2 \coloneqq 0.001.
	\end{align*}
Following the definition of an active variable, we can directly read the active set from the function definition which are respectively
\begin{equation*}
	\mathcal{A}_{f^1} = \set{0,1,8,14}, \mathcal{A}_{f^2} = \set{0,1,5,8,14}.
\end{equation*}
Applying formally Algorithm~\ref{alg:det_active_set}, the plot of the Kolmogorov-Smirnov distance of the weighted samples along each dimension, shows that the variables, whose indices are elements of $\set{0,1,8,14}$ are non uniform distributed for the first density function $f^1$ and the variables with index in  $\set{0,1,8,5,14}$ are non uniformly distributed by the second function $f^2.$ Since all variables are uncorrelated as the correlation it shows.
	\begin{figure}[h]
		\centering
		\includegraphics[width=0.7\textwidth]{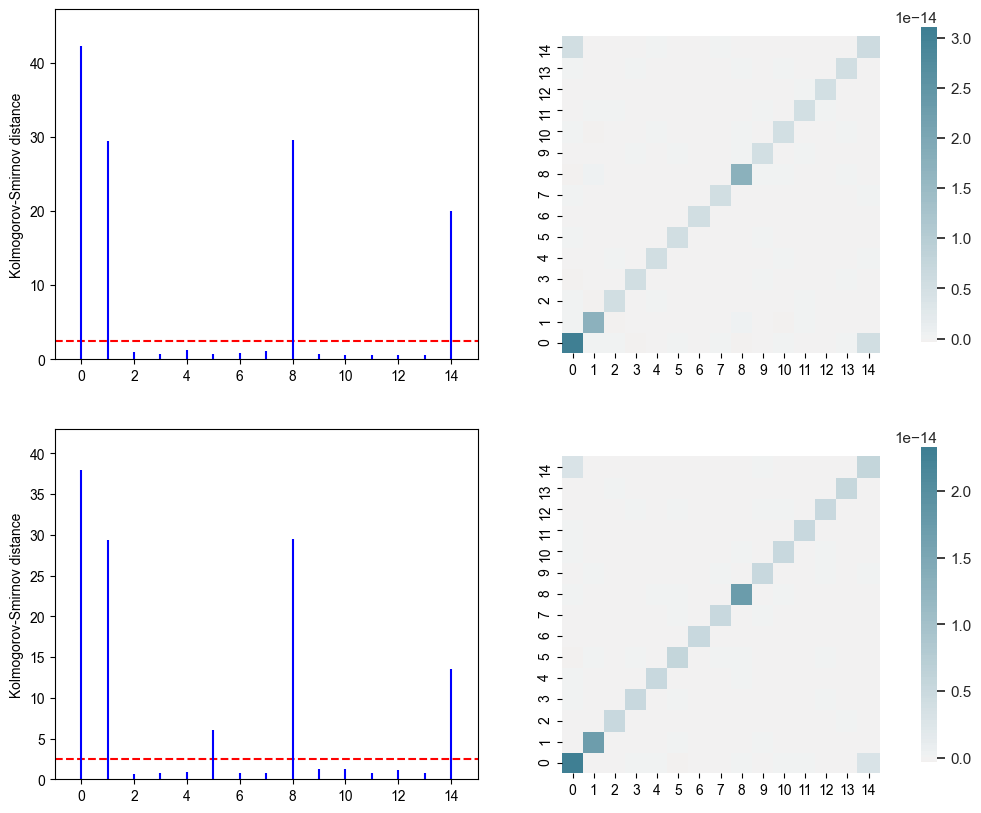}
	\caption{Kolmogorov-Smirnov distance(left) of each univariate samples $(x^n_i, w^n)_{n \in [N]}, N=10^4,  i=0, \ldots, 14.$  and correlation coefficient of the different variables $\text{Cor}(X_i, X_j), i,j=1, \ldots, 15$ (right). Top: Test function $f^1$ from example~\ref{ex:theor_ex_1}. Bottom: Test function $f^2.$ }
	\label{fig:active_set}
\end{figure}
\end{example}

Assuming that the active set $\mathcal{A}_f$  of the sparse mixture model is already known and increasingly ordered, we can iterate over the index $i \in \mathcal{A}_f$ of the active variables to determine the marginal distribution of the subset $X_{[i]}$ of the random variable $X$ with probability density function $f$. For shake of simplicity we will Denote by $P_{[i]}f$ the marginal probability density function with respect to $X_{[i]}$. It is equal to the marginal density function with respect to $X_{\bar{i}},$ where $\bar{i} \coloneqq [i] \cap \mathcal{A}_f.$ Let $r(i)= \abs{\bar{i}}$ represents the position of $i$ in $\mathcal{A}_f$ and by the same way the number of iterations. Theorem~\ref{thm: marginal_mixture_model} implies that for each $i \in \mathcal{A}_f$ the index set of all interacting variables of the marginal mixture model with respect to $X_{\bar{i}}$ is equal to 
\begin{equation}
	U^{r(i)} = \set{u^{r(i)}:=u\cap \bar{i}\mid u \in U},
\end{equation}  
with parameters set
\begin{equation*}
	\theta^{r(i)} = \set{\theta_{u_k\cap v} \coloneqq (\theta_j)_{j\in u_k\cap v}, k=1,\ldots,K \mid u_k \in U \text{ and } v=u^{r(i)}_k \in U^{r(i)}}.
\end{equation*}
%The $\abs{\bar{i}}$-dimensional marginal density function with respect to $X_{\bar{i}}$ can be rewritten
%\begin{equation*}
	%	P_{\bar{i}} f(x\mid \alpha, \theta)\coloneqq f_i(x\mid \alpha^{r(i)}, \theta^{r(i)})\coloneqq \sum_{v \in U^{r(i)}} \alpha_{v}^{r(i)} p(x_{v} \mid \theta_{v}^{r(i)}),
	%\end{equation*}
%where $p(\cdot \mid \theta_{v}^{r(i)})$ is the marginal probability density function with parameter $\theta_v^{r(i)} \in \theta^{r(i)}$ on $\T^{\abs{v}}.$ The mixture weights yield
%\begin{equation}
	%	\alpha_{v}^{r(i)}:= \sum\limits_{\substack{u_k \in \mathcal{T}(U, v,i)\\ k=1,\ldots, K}} \alpha_{u_k},
	%\end{equation}
%with
%\begin{equation*}
	%	\mathcal{T}(U, v,i) := \set{u \in U \mid u\cap \bar{i} = v \text{ and } \theta_u^{r(i)} =\theta_v},
	%\end{equation*}
%and $P_{[i]}: L^1(\T^{d_i}) \rightarrow  L^1(\T^{d_i}), d_i = \abs{\mathcal{A}_{f_i}}$ denotes the \textit{truncation operator}. 
Note that for $u^{r(i)}=\set{},$ the function $p(\cdot \mid \theta_{v}^{r(i)})$ is the probability density function of the uniform distribution, and for $r(i)=\abs{\mathcal{A}_f},$ the marginal density function $P_{[i]}f \coloneqq f_i$ is equal to $f$. By definition of the marginal mixture model, it follows that the active set $\mathcal{A}_{f_i} \subseteq \mathcal{A}_f$ and $\mathcal{A}_{f_i}^{u_k} \subseteq \mathcal{A}_f^{u_k}$ for all $k \in [K]$ and $i\in \mathcal{A}_f$. Thus we can further define for a fix $i \in \mathcal{A}_f$ the \textit{residual active set} of the ground function $f$ with respect to the marginal density function by
\begin{equation*}
	\overline{\mathcal{A}_f^{r(i)}} \coloneqq \mathcal{A}_f \setminus \mathcal{A}_{f_i},
\end{equation*}
and the residual active set of each mixture component density function $ p(\cdot \mid  \theta_{u_k}), u_k \in U, k=1, \ldots, K$ with respect to its marginal counterpart by
\begin{equation*}
	\overline{\mathcal{A}_f^{r(i),u_k}} \coloneqq \mathcal{A}_f^{u_k} \setminus \mathcal{A}_{f_i}^w,
\end{equation*}
where $\theta_w=\theta_{u_k}^{r(i)}, w \in U^{r(i)}$.
This notion of residual active set will be useful to considerably reduced the complexity of the algorithm presented in section~\ref{sec:learning}, which approximate iteratively the marginal density function of the sparse mixture model.
These new notions can be illustrated by two concrete examples. Indeed in the following we will consider two $15$-dimensional density function of sparse wrapped Gaussian mixture models. We will compute their marginal distributions with respect to the subset of random variables $X_u, \emptyset \neq u \subset [d].$  

%\begin{example}
%	The third test function $f^3$ has the parameters 
%	\begin{align*}
%		U&\coloneqq\set{u_1,u_2, u_3 ,u_4} = \set{\set{0,14}, \set{4,5,6}, \set{14}, \set{14}} ,\\ 
%		\alpha
%		&\coloneqq(0.3,0.2, 0.2, 0.3),\\
%		\mu 
%		&\coloneqq \frac12 \left((1,1)^{\intercal}, (1,1,1)^{\intercal},1, 0.4 \right),\\
%	\end{align*}
%	and for  $\sigma^2=0.001$ the covariance matrices are defined as
%		$$
%	\Sigma_1 \coloneqq \sigma^2\Bigg(\begin{array}{ccc}1&0.3&0.2\\0.3&1&0.1\\0.2&0.1&1\end{array}\Bigg),
%	$$
%	and
%	$$
%	\Sigma_2 \coloneqq \sigma^2\Big(\begin{array}{cc}1&0.5\\0.5&1\end{array}\Big), \quad 
%	\Sigma_3=\Sigma_4=\sigma^2.
%	$$
%\end{example}

In the following we will introduce two notions of \textit{effective dimension}, when dealing with very high dimensional sparse functions.
\begin{definition}\label{def:effective_dimension}
	Let $f: \T^d \rightarrow \R$ be a function and $\eta \in (0,1]$. The \textit{superposition dimension} $d_s>0$ at level $\eta$ is defined as
	\begin{equation}\label{eq:superpos_theshold}
		d_s := \argmin_{s \in [d-1] }\frac{\sum_{\substack{\emptyset \neq u \subseteq [d] \\ \abs{u} \leq s}}\sigma^2(f_u)}{\sigma^2(f)} \geq \eta,
	\end{equation}
	where the $\abs{u}$-dimensional functions  
	\begin{equation*}
		f_u = \sum_{v \subseteq u}(-1)^{\abs{u}-\abs{v}} P_vf
	\end{equation*}
	are the ANOVA-term of the function $f, P_v$ denotes the projection operator of definition~\ref{def:marginal_density_function} and $\sigma^2(\cdot)$ the variance of the corresponding functions. 
	%For $f$ being the probability density function of a $d$-dimensional random variable $X$, the operator $P_vf$ represents the marginal density function with respect to $X_v,$ for all $v \subseteq [d]$.
	The second notion of \textit{effective dimension} is the \textit{ truncation dimension }, which is defined as 
	\begin{equation*}
		d_t \coloneqq \argmin_{s \in [d-1] }\frac{\sum_{\emptyset \neq u \subseteq [s]} \sigma^2(f_u)}{\sigma^2(f)} \geq \eta.
	\end{equation*}
\end{definition}
We will combine later in section~\ref{sec:num_approx} these two notions of effective dimension to introduce an assumption of sparsity criterion for the density functions, we want to approximate. First, the function $f: \T^d \rightarrow \R$ from \eqref{ex:theor_ex_1} can be rewritten as
\begin{equation*}
	f=\sum_{w\in W} g_w,
\end{equation*}
where $W=U \subset \mathcal{P}([d]),$ and $g_w$ are linear combination of lower-dimensional functions depending only on variables with index set in $w$. For those class of density functions, it has been shown in \cite[proposition~2.1]{hertrich2021sparse}, that the ANOVA-decomposition of $f$ is equal to
\begin{equation*}
	f= \sum_{u \in \bar{W}} f_u,
\end{equation*} 
where $\bar{W}$ denotes the set of all $w \in W$ and all their subsets. Then the superposition dimension defined in \eqref{eq:superpos_theshold} is also 
\begin{equation*}
	d_s = \argmin_{s \in [d-1] } \frac{\sum_{\substack{\emptyset \neq u \in \bar{W}\\ \abs{u} \leq s}}\sigma^2(f_u)}{\sigma^2(f)},
\end{equation*}
and the truncation dimension
\begin{equation*}
	d_t = \argmin_{s \in [d-1]} \frac{\sum_{\emptyset \neq u \in \bar{W}_s}\sigma^2(f_u)}{\sigma^2(f)},
\end{equation*}
where 
\begin{equation*}
	\bar{W}_s \coloneqq \set{u \subseteq [d] \mid u \in \bar{W} \text{ and } u\subseteq [s] }.
\end{equation*}
Considering the ANOVA decomposition of the marginal density function $P_{\xi}f$ of $f$ associated to an arbitrary but fixed $i \in [d-1], \xi \coloneqq [i],$  it follows that 
\begin{equation*}
	P_{\xi}f = \sum_{u \subseteq [d]} \left(P_{\xi}f\right)_u, 
\end{equation*}
where 
\begin{equation}\label{eq:anova_term}
	\left(P_{\xi}f\right)_u = \sum_{v \subseteq u} (-1)^{\abs{u}-\abs{v}} P_vP_{\xi}f.
\end{equation}
We know by definition, that $P_{\xi}f$ depends only on the variables $x_{\xi}.$ Lemma~\cite[proposition~2.1]{hertrich2021sparse} implies that $(P_{\xi}f)_u =0$ for all $u \subseteq [d],$ such that $u$ is not included in $\xi.$ Thus
\begin{equation*}
	P_{\xi}f = \sum_{u \subseteq \xi} \left(P_{\xi}f\right)_u. 
\end{equation*}
For $v \subseteq u \subseteq \xi$ it holds that $P_vP_{\xi}f =P_vf$. Hence equation~\eqref{eq:anova_term} implies that $\left(P_{\xi}f\right)_u = f_u$ for all $u \subseteq \xi$. Thus with \cite[Lemma~2.9]{Potts2019} the truncation dimension defined in \ref{def:effective_dimension} yields
\begin{equation*}
	d_t = \argmin_{s \in [d-1]}\frac{\sigma^2(P_{[s]}f)}{\sigma^2(f)}.
\end{equation*}
Using this, we can introduce an iterative algorithm, which can approximate the marginal density function $P_{[i]}f$ for any $i \in \mathcal{A}_f$ of the form \eqref{eq:sparse_density_function}, under the assumption that a large enough number of samples are provided. The function $f$ is an accurate approximation to the ground function $g$. If $f$ is sparse in sense of equation~\ref{eq:sparse_density_function}, then there exists an element $v \subset [d], \abs{v}\ll d$ such that
\begin{equation*}
	\norm{g- P_{v}f} \leq \norm{g-f} + \norm{f_{X_{v^c}}-\mathbf{I}}\norm{P_{v}f}  \rightarrow 0,
\end{equation*}
and the maximal number of interacting variables are very small with respect to the space dimension.

%-------------------------------------------
\section{Learning Sparse Mixture Models} \label{sec:learning}

In the rest of this paper, we assume that all variables $x_1,\ldots,x_d$ in \eqref{eq:sparse_density_function}
are active.
For learning the sparse MM, we propose an algorithm which iteratively approximates the marginals
$\int_{\mathbb T^{d-r}} f \dx x_{r+1} \ldots \dx x_d$ for $r=1,\ldots,d$.
In the following, we give an idea of the algorithm  by describing its first two steps.
Let the samples $\{x^n = (x_k^n)_{k=1}^d: n =1,\ldots,N\}$ be given.
\\[2ex]
\textbf{Step 1:} Find an approximation of the first marginal by
\begin{equation} \label{m1}
	f^1(x_1) = \alpha_0^1 + \sum_{k=1}^{K_1} \alpha_k^1 p(x_1|\mu_k^1,\sigma_k^1) 
\end{equation}
from the samples $\{x_1^n:n =1,\ldots,N\}$ as follows:
\begin{itemize}
	\item[1.1] Determine $K_1$ by the BIC method described in Appendix \ref{sec:BIC}.
	\item[1.2] Apply a univariate EM algorithm 
	to compute $(\alpha_k^1,\mu_k^1,\sigma_k^1)$, $k=1,\ldots,K_1$ and $\alpha_0^1$ and to
	determine the probability $\beta^1_{k,n}$, $n  =1,\ldots,N$, $k=0,\ldots,K_1$  
	that $x_1^n$ belongs to the $k$-th mixture component.
\end{itemize}

\textbf{Step 2:} Find an approximation of the first two marginals by the following steps:
\begin{itemize}
	\item[2.1] For each $k=0,\ldots,K_1$ determine if the weighted samples
	$\{ \beta^1_{k,n} x_2^n: n =1,\ldots,N\}$ are uniformly distributed and uncorrelated
	by the Kolmogorov-Smirnov test in Appendix \ref{sec:ks} and correlation estimate in Appendix \ref{sec:ct}.
	Then we get
	$$
	\{0,\ldots,K_1\} = K_{nu} \cup K_{u},
	$$
	where $K_{nu}$ denote the indices of those mixture summands in \eqref{m1}, 
	where the samples are not uniformly distributed
	and $K_{u}$ the other ones.
	
	\item[2.2] For each $k \in K_{nu}$ and samples $\{ \beta_{n,k}^1 x_2: n  =1,\ldots,N \}$
	determine
	\begin{equation} \label{m2}
		f_k^2(x_2) =  \alpha_{k,0}^2 + \sum_{l=1}^{L_k} \alpha_{k,l}^2 p(x_2|\mu_{k,l}^2,\sigma_{k,l}^2) 
	\end{equation}
	by computing
	\begin{itemize}
		\item $L_k$ by the BIC method described in Appendix \ref{sec:BIC}.
		\item $(\alpha_{k,l}^2,\mu_{k,l}^2,\sigma_{k,l}^2)$, $l=1,\ldots,L_k$ and $\alpha_{k,0}^2$
		by a univariate EM algorithm. These parameters will be used as initial ones in the next EM step.
	\end{itemize}
	\item[2.3] Case 1: If $0 \in K_u$, set $p(x_1|\mu_{0}^{1},\sigma_{0}^{1}) := 1$ and 
	compute the parameters $(\alpha_{k,l}^{1,2}, \mu_{k,l}^{1,2}, \Sigma_{k,l}^{1,2})$ 
	and 
	determine the probability $\beta^{1,2}_{kl,n}$, $n  =1,\ldots,N$, $l=1,\ldots,L_k$, $k\in K_{nu}$,
	in the MM
	\begin{align} \label{m12_1}
		f_k^{1,2}(x_1,x_2) &=   \alpha_0^1 +
		\sum_{k \in K_{u}} \alpha_{k}^{1} p(x_1|\mu_{k}^{1},\sigma_{k}^{1}) \\
		&+ \sum_{k \in K_{nu}} 
		\Big( \sum_{l=1}^{L_k} \alpha_{k,l}^{1,2} p(x_1,x_2|\mu_{k,l}^{1,2},\Sigma_{k,l}^{1,2}) 
		+ \alpha_{k,0}^{1,2} p(x_1|\mu_{k}^{1},\sigma_{k}^{1}) \Big) \nonumber
	\end{align}
	with $\alpha_{k,0} \coloneqq \alpha_k^1 \alpha_{k,0}^2$ and initialization for $l=1,\ldots,L_k$ as
	\begin{equation} \label{init}
		\alpha _{k,l}^{1,2} \coloneqq \alpha_k^1 \alpha_{k,l}^2, \quad
		\mu_{k,l}^{1,2} \coloneqq \begin{pmatrix} \mu_k^1\\ \mu_{k,l}^2 \end{pmatrix}, \quad
		\Sigma_{k,l}^{1,2} \coloneqq 
		\begin{pmatrix} \sigma_k^1& 0\\ 0& \sigma_{k,l}^2 
		\end{pmatrix}.
	\end{equation}
	Case 2: If $0 \in K_{nu}$, 
	compute the parameters $(\alpha_{k,l}^{1,2},\mu_{k,l}^{1,2}, \Sigma_{k,l}^{1,2})$
	and determine the probability $\beta^{1,2}_{kl,n}$, $n  =1,\ldots,N$, $l=1,\ldots,L_k$, $k\in K_{nu}$
	in the MM
	\begin{align} \label{m12_2}
		f_k^{1,2}(x_1,x_2) 
		&=  \alpha_{0,0}^{1,2} + \sum_{l=1}^{L_0} \alpha_{0,l}^{1,2} p(x_2|\mu_{0,l}^{1,2},\sigma_{0,l}^{1,2}) +
		+ \sum_{k \in K_{u}} \alpha_{k}^{1} p(x_1|\mu_{k}^{1},\sigma_{k}^{1}) \\
		&+ \sum_{k \in K_{nu}\setminus \{0\}} \Big(
		\sum_{l=1}^{L_k} \alpha_{k,l}^{1,2} p(x_1,x_2|\mu_{k,l}^{1,2},\Sigma_{k,l}^{1,2}) 
		+ \alpha_{k,0}^{1,2} p(x_1|\mu_{k}^{1},\sigma_{k}^{1}) 
		\Big)
		\nonumber
	\end{align}
	We use the same initialization \eqref{init} for $k \in K_{nu} \setminus \{0\}$ and
	\begin{equation} \label{init_0}
		\alpha _{0,l}^{1,2} \coloneqq \alpha_0^1 \alpha_{0,l}^2, \quad
		\mu_{0,l}^{1,2} \coloneqq  \mu_{0,l}^2, \quad
		\sigma_{0,l}^{1,2} \coloneqq 
		\sigma_{0,l}^2 .
	\end{equation}
\end{itemize}
If we use a MM with wrapped Gaussians with just diagonal covariance matrices, Step 2.3 is superfluous
and the new parameters are those from the initialization.
\begin{remark}
	If we consider the sparse mixture model of diagonal wrapped Gaussian or von Mises distribution then the estimation step $12$ and $18$ in algorithm~\ref{alg:density_func_approx} will be resumed to fitting univariate marginal distribution. This will considerably increase the computation (time and storage) complexity.
\end{remark}
\begin{algorithm}[!ht]
	\caption{Active set detection and parameters estimate}\label{alg:density_func_approx}
	\begin{algorithmic}[1]
		\State \textbf{Input}: $(x^1,...,x^N)\in\T^{d, N}$, $(w_1,\ldots,w_N) \in \mathbb R^N$, $K_{max}\in \mathbb{N}, \varepsilon_{KS} ,\varepsilon_{c}>0$
		\State \textbf{Output}: $U, \hat{\theta}$
		\State \textbf{Set} $U=\set{u=\set{}}, \alpha=(1),\theta=\set{()}, r(-1)=0, n_u=1$
		%\State Compute the active set $\mathcal{A}_f=\mathcal{A}_{f_{-1}}$ of the target $f$ using algorithm~\ref{alg:det_active_set}
		%\begin{align}
		%	\mathcal A_{f}=\mathcal A_{f_{-1}}&=\left\lbrace i \in \mathcal [d] \mid D_N\left((x^n_i, w_n)_{n\in [N]}\right)\geq \varepsilon_{KS} \text{ or } \right.\\ 
		%		 &\hspace{4cm} \left. \text{Cor}(X_i, X_j) \geq \varepsilon_c, \forall i \neq j \right\rbrace
		%\end{align} 
		\For {$i=1, \ldots d$}
		\State \textbf{Set} $K = \sum_{u \in U^{(i-1)}}n_u, U^{(i)}_{new}=\set{}, U^{(i)}_{fix}=\set{}, \theta^{(i)}_{new}=\set{}, \theta^{(i)}_{fix}=\set{}$
		\For{$k=0,\ldots, K$}
		\State Update the samples weights w.r.t the $k$-component of $P_{[i-1]}f$
		\begin{equation*}
			\tilde{w}_{n,k}^i = w_n\beta_{n,k}^{i}, \text{ where }	\beta_{n,k}^{i} = \frac{\alpha_k p(x^n_{u_k}\mid \theta_k)}{\sum_{t=1}^K \alpha_t p(x^n_{u_t}\mid \theta_t)},
		\end{equation*}
		%\qquad\quad where $u_k, u_t \in U^{r(i)-1}.$ 
		\State Determine the corresponding residual active set
		\begin{equation}
			\overline{\mathcal A_{f_i}^{u_k}} = \set{ i \in \overline{\mathcal A_{f_{s(i)}^{u_k}}} \mid  D_N\left((x^n_i, \tilde{w}_{n,k}^i)_{n\in [N]}\right)\geq \varepsilon_{KS} }
		\end{equation}
		\If{ $i \in \overline{\mathcal A_{f_i}^{u_k}}$}
		\State Add to $U^{(i/2)}_{new}\leftarrow u^{(i/2)} = u_k \cup \set{i}$
		\State Determine the number of the mixture components with the index set $u_k^{(i)}$
		\begin{align*}
			n_k &= \argmin_{k=1,\cdots, K_{max}} \text{BIC}\left((x^n_i, \tilde{w}_{n,k}^i), k\right)\\
			(\alpha_j,\theta_j^{(i)})_{j \in [n_k]}&=\text{Prox-EM}\left((x_i^n, \tilde{w}^n)_{n\in [N]},(\theta^{0}_t)_{t \in [n_k]}\cup \theta_{\emptyset},\gamma_2 \right)
		\end{align*}
		\State Set $\theta = \set{(\theta_{u_k}^{(i-1)}, \theta_j^{(i)}), \alpha_k\cdot \alpha_j}_{j \in [n_k]}$ and add $\theta^{(i/2)}_{new}\leftarrow \theta$
		\Else 
		\State Add $U^{(i/2)}_{fix}\leftarrow u_k$ and $\theta^{(i)}_{fix} \leftarrow \theta_{u_k}^{(i-1)}$
		\EndIf
		\EndFor
		\State Compute $\bar{w}_n = w_n \cdot \beta_{n,G_2}$ the posterior probabilities $\beta_{n,G_2}$ of the samples according to section~\ref{sec:em}, where $G_2 = \left(U^{(i/2)}_{new},\theta^{(i/2)}_{new}\right)$ and set $\bar{w}_n = w_n \cdot \beta_{n,G_2},$ for all $n\in [N].$ 
		\State Update the parameters with the EM-Algorithm~\ref{alg:alg_em_mm_orig}
		\begin{equation*}
			U^{(i)}_{new},\alpha^{(i)}_{new}, \theta^{(i)}_{new} = \text{Prox-EM}\left((x^n, \bar{w}_n)_{n\in [N]}, \theta^{(i/2)}_{new}, \gamma_1\right)
		\end{equation*}
	\State Set $U^{r(i)}=U^{(i)}_{new}\cup U^{(i)}_{fix}$ and $\theta^{(i)}=\theta^{(i)}_{new}\cup \theta^{(i)}_{fix}$
		\EndFor
	\end{algorithmic}
\end{algorithm}
%---------------------------------------------------------------------
\section{Experimental Results} \label{sec:numerics}
In the following we will apply  Algorithm~\ref{alg:density_func_approx} to determine the collection of variable interactions and the associated mixture components parameters for the test functions defined in \ref{sec:}, the product of B-splines function and the California Housing data. To evaluate the model the log-likelihood the training data and the test will be compared with each other. Furthermore the relative $L^p, p=1,2$ errors between the ground truth function $f$ and the approximated model  $\hat{p}$ on unknown test data will be computed. Recall that the relative $L^p$-error is defined as
\begin{equation*}
	e_{L^p}(\hat{p}, f)=\frac{\norm{\hat{p}-f}_{L^p}}{\norm{f}_{L^p}},
\end{equation*}
  where  $\norm{\cdot}_{L^p}$ will be determined via the Monte-Carlo integration, i.e
  \begin{align}
  	\norm{f}_{MC}=\frac{1}{N_{MC}}\sum_{n=1}^{N_{mc}}\abs{f(x^n)}^p  \text{ and }  \norm{\hat{p}-f}_{MC}=\frac{1}{N_{MC}}\sum_{n=1}^{N_{mc}}\abs{\hat{p}(x^n)-f(x^n)}^p
  \end{align}
where $\left(x^n\right)_{n\in N_{MC}}$ are uniformly distributed samples on $\T^d$. For large value of $N_{MC}$ the Monte-Carlo norm $\norm{\cdot}_{MC}$ yields an accurate approximation of $\norm{\cdot}_{L^p}.$ Therefore $N_{MC}=10^5$ will be used to compute the relative errors. 
  
  To train the function from example\ref{ex:theor_ex_1} $N_{tr}=75\cdot 10^2,10^4$ training samples has been drawn by the rejection sampling method\cite{hertrich2021sparse, WANG2000231}. The parameters used by Algorithm~\ref{alg:density_func_approx} to learn the model are in Table~\ref{tab:Hyperparameter_approx_test_functions}. The approximated model from Algorithm~\ref{alg:density_func_approx} given in Table~\ref{tab:param_estimate_example_wrapped_normal_test_function} are obviously  an accurate approximation of the ground truth mixture models. The negative log-likelihood and the relative error in Table~\ref{tab:log_likeli_rela_parame} prove this. 
  
    %Figure~\ref{} shows the marginal density functions with respect to $X_u, u\in \set{\set{1}, \set{1,2}}.$ Choosing $d_t=d,$ one can see that the approximation errors with respect to the $L_1$ and $L_2$ norm are much better if we compare our algorithm with the Kernel density function and when using an EM algorithm without any sparsity assumption on the ground function.

\begin{table}
	\begin{center}
		\resizebox*{!}{4.7cm}{
			\begin{tabular}{|c|c|c|c|c|c|c|}
				$N$ samples $\cdot 10^2$&Method&Truth&$\varepsilon_{KS}$&$\varepsilon_{c}$&$\gamma_1 \cdot 10^{-3}$&$\gamma_2 \cdot 10^{-3}$\\
				\hline
				\multirow{4}{*}{$75$}&\multirow{2}{*}{wrapped Gaussian}&a)&\multirow{4}{*}{$4.7$}&\multirow{8}{*}{$0.1$}&\multirow{2}{*}{$3.0$}&\multirow{8}{*}{$1.0$}\\
				&&b)&&&&\\ \cline{2-3}
				&\multirow{2}{*}{von Mises}&a)&&&\multirow{2}{*}{$5.0$}&\\
				&&b)&&&&\\ \cline{1-4}
				\multirow{4}{*}{$100$}&\multirow{2}{*}{wrapped Gaussian}&a)&\multirow{2}{*}{$5.0$}&&\multirow{2}{*}{$3.0$}&\\ 
				&&b)&&&&\\ \cline{2-3}
				&\multirow{2}{*}{von Mises}&a)&&&\multirow{2}{*}{$5.0$}&\\
				&&b)&&&&\\
				\hline
			\end{tabular}
		}\\
	\end{center}
	\caption{Hyper-parameters for the first two test functions}
	\label{tab:Hyperparameter_approx_test_functions}
\end{table}
\begin{table}
	\begin{center}
		\resizebox*{!}{2.7cm}{
			\begin{tabular}{c|c|cccc}
				Truth&Method&$\mathcal L_f(x^1,...,x^N)$&$\mathcal L_{\hat p}(x^1,...,x^N)$&$e_{L^1}(\hat p,f)$&$e_{L^2}(\hat p,f)$\\\hline
				a)&wrapped&$36546.7\pm 190.5$&$36553.7\pm 189.3$&$0.0343\pm 0.0069 $&$0.0312\pm 0.0059$\\
				a)&comp.\ wrapped&$36546.7\pm 190.5$&$36551.9\pm 189.5$&$0.0329\pm 0.0076$&$0.0311\pm 0.0069$\\
				a)&von Mises&$36588.8\pm 201.4$&$36592.7 \pm 199.8$&$0.0325\pm0.0072 $&$0.0312\pm 0.0123$\\\hline
				b)&wrapped&$33591.0\pm 261.2$&$33600.9\pm 258.2$&$0.0371\pm 0.0061$&$0.0302\pm 0.0073$\\
				b)&comp.\ wrapped&$33591.0\pm 261.2$&$33598.7\pm 258.4$&$0.0356\pm 0.0079$&$0.0304\pm 0.0097$\\
				b)&von Mises&$33647.5\pm 274.5$&$33514.7\pm 498.4$&$0.0505\pm 0.0501$&$0.0285\pm 0.0094$
			\end{tabular}
		}\\
		\vspace{0.3cm}
		\resizebox*{!}{2.7cm}{
			\begin{tabular}{c|c|cccc}
				Truth&Method&$\mathcal L_f(x^1,...,x^N)$&$\mathcal L_{\hat p}(x^1,...,x^N)$&$e_{L^1}(\hat p,f)$&$e_{L^2}(\hat p,f)$\\\hline
				a)&wrapped&$48660.7\pm 176.3 $&$48669.5\pm 175.1$&$0.0297\pm 0.0080 $&$0.0239\pm 0.0077$\\
				a)&comp.\ wrapped&$48660.7\pm 176.3$&$48667.6\pm 175.0$&$0.0299\pm 0.0073$&$0.0275\pm 0.0074$\\
				a)&von Mises&$48763.3\pm 189.8$&$48768.0\pm 188.0 $&$0.0287\pm 0.0081$&$0.0231\pm 0.0083$\\\hline
				b)&wrapped&$44967.1\pm253.7$&$44977.8\pm252.5$&$0.0403\pm0.0117$&$0.0344\pm0.0179$\\
				b)&comp.\ wrapped&$44967.1\pm 253.7$&$44975.4\pm252.5$&$0.0385\pm0.0120$&$ 0.0312\pm0.0152$\\
				b)&von Mises&$44881.0\pm308.1$&$44512.3\pm921.8$&$ 0.0720\pm0.0981$&$0.0373\pm 0.0247$
			\end{tabular}
		}
	\end{center}
	\caption{Approximation of $f^i, i=1,2$ in example~\ref{ex:theor_ex_1} by  sparse mixture models of wrapped Gaussian and von Mises distribution from section~\ref{sec:SPAMM}. 
		Top: $N=7500$, Bottom: $N=10000$. 
		Average value of the log likelihood function and relative $L_q$ errors, $q=1,2$ for $N_{te}=10^5$ training sets
		top: $N_{tr}=7500,$  bottom: $N_{tr}=10000$.}
	\label{tab:log_likeli_rela_parame}
\end{table}

\begin{table}
	\begin{center}
		\resizebox*{!}{8.7cm}{
			\begin{tabular}{c|c|c|cccl}
				Truth&Method&Time(in $s$)&$\hat{u}$&$\hat{\alpha}$&$\hat{\mu}$&$\hat{\Sigma}/\hat{\kappa}$\\\hline
				\multirow{2}{*}{a)}&\multirow{2}{*}{wrapped}&\multirow{2}{*}{$18.2\pm 2.41$}&$\set{0,1,8} $&$0.7018$&$\begin{pmatrix}
				0.50045889\\ 0.49997703\\ 0.50017948
				\end{pmatrix}$&$\begin{pmatrix}
					9.8064\cdot10^{-4} &2.9304\cdot10^{-5} & 1.3477\cdot10^{-5} \\
					2.9304\cdot10^{-5} & 9.8954\cdot10^{-4} &-2.0069\cdot10^{-5} \\
					1.3477\cdot10^{-5} & -2.0069\cdot10^{-5} & 1.0160\cdot10^{-3}
				\end{pmatrix}$\\\cline{4-7}
				&&& $\set{0,14}$&$0.2982$&$\begin{pmatrix}
				0.24928052 \\0.24895061
				\end{pmatrix}$&$\begin{pmatrix}
					1.0302\cdot10^{-3}& 2.3457\cdot10^{-6} \\
					2.3457\cdot10^{-6} & 9.7430\cdot10^{-4}
				\end{pmatrix}$\\\hline
				\multirow{2}{*}{a)}&\multirow{2}{*}{comp.\ wrapped}&\multirow{2}{*}{$17.3\pm 2.6$}&$\set{0,1,8} $&$0.7019$&$\begin{pmatrix}
					0.50046016\\ 0.49997673\\ 0.50017988
				\end{pmatrix}$&$\begin{pmatrix}
					9.9049\cdot10^{-4} &0 &0\\
					0&9.9954\cdot10^{-4} &0\\
					0&0&1.0259\cdot10^{-3}
				\end{pmatrix}$\\\cline{4-7}
				&&& $\set{0,14}$&$0.2981$&$\begin{pmatrix}
				0.24928027 \\0.24895058
				\end{pmatrix}$&$\begin{pmatrix}
					1.0401\cdot10^{-3} & 0\\
					0&9.8431\cdot10^{-4}
				\end{pmatrix}$\\\hline
				\multirow{2}{*}{a)}&\multirow{2}{*}{von Mises}&\multirow{2}{*}{$56.7\pm 5.9$}&$\set{0,1,8}$&$0.6977$&$\begin{pmatrix}
				0.5003044\\  0.50002035\\ 0.49994604
				\end{pmatrix}$&$\begin{pmatrix}
					25.1364 \\ 25.4775 \\ 26.1858
				\end{pmatrix}$\\\cline{4-7}
				&&& $\set{0,14}$&$0.3023$&$\begin{pmatrix}
					0.25072411\\ 0.2506996
				\end{pmatrix}$&$\begin{pmatrix}
				25.4839 \\  26.8531
				\end{pmatrix}$\\\hline
				\multirow{3}{*}{b)}&\multirow{2}{*}{wrapped}&\multirow{2}{*}{$26.0\pm 2.8$}&$\set{0,1,8} $&$0.7067$&$\begin{pmatrix}
				0.49912957\\ 0.49998626\\ 0.50040215
				\end{pmatrix}$&$\begin{pmatrix}
					9.9217\cdot10^{-4} &-1.2127\cdot10^{-5} & -2.4808\cdot10^{-6}\\
					-1.2127\cdot10^{-5} & 1.0013\cdot10^{-3} &-1.7085\cdot10^{-5}\\
					-2.4808\cdot10^{-6} &-1.7085\cdot10^{-5}  &9.8552\cdot10^{-4}\\
				\end{pmatrix}$\\\cline{4-7}
				&&& $\set{0,14}$&$0.1955$&$\begin{pmatrix}
				0.24912379\\ 0.25057763
				\end{pmatrix}$&$\begin{pmatrix}
				9.8932\cdot10^{-4} & -7.8356\cdot10^{-6}\\
				-7.8356\cdot10^{-6} & 1.0164\cdot10^{-3}
				\end{pmatrix}$\\\cline{4-7}
				&&& $\set{5}$&$0.0978$&$0.29812577$&$1.0603\cdot10^{-3}$\\\hline
				\multirow{3}{*}{b)}&\multirow{3}{*}{comp.\ wrapped}&\multirow{3}{*}{$ 24.2\pm 2.3$}&$\set{0,1,8} $&$0.7014$&$\begin{pmatrix}
				0.50035103\\ 0.49968499 \\0.49978892
				\end{pmatrix}$&$\begin{pmatrix}
					1.0204\cdot 10^{-3} &0 & 0\\
					0 &1.0274\cdot 10^{-3} & 0\\
					0 & 0& 1.0058\cdot 10^{-3}
				\end{pmatrix}$\\\cline{4-7}
				&&& $\set{0,14}$&$0.2009$&$\begin{pmatrix}
				0.24977141\\ 0.24948918
				\end{pmatrix}$&$\begin{pmatrix}
					1.0074\cdot 10^{-3} & 0\\
					0& 9.7733\cdot 10^{-4}
				\end{pmatrix}$\\\cline{4-7}
				&&& $\set{5}$&$0.0976$&$0.30006964$&$1.0835\cdot 10^{-3}$\\\hline
				\multirow{3}{*}{b)}&\multirow{3}{*}{von Mises}&\multirow{3}{*}{$74.8\pm 18.9$}&$\set{0,1,8} $&$0.7030$&$\begin{pmatrix}
					0.49973165 \\0.50038246 \\0.4999259
				\end{pmatrix}$&$\begin{pmatrix}
					25.9242\\ 25.6306\\ 25.6864
				\end{pmatrix}$\\\cline{4-7}
				&&& $\set{0,14}$&$0.1968$&$\begin{pmatrix}
					0.25210973\\ 0.24947434
				\end{pmatrix}$&$\begin{pmatrix}
					25.5994\\ 26.2484
				\end{pmatrix}$\\\cline{4-7}
				&&& $\set{5}$&$0.1001$&$\begin{pmatrix}
				0.2996465
				\end{pmatrix}$&$25.9420$\\\hline
			\end{tabular}
		}\\
		\vspace{0.3cm}
		\resizebox*{!}{8.7cm}{
			\begin{tabular}{c|c|c|cccl}
				Truth&Method&Time(in $s$)&$\hat{u}$&$\hat{\alpha}$&$\hat{\mu}$&$\hat{\Sigma}/\hat{\kappa}$\\\hline
				\multirow{2}{*}{a)}&\multirow{2}{*}{wrapped}&\multirow{2}{*}{$16.2\pm 2.6 $}&$\set{0,1,8} $&$0.7016$&$\begin{pmatrix}
				0.49992579\\ 0.50091679 \\0.50041143
				\end{pmatrix}$&$\begin{pmatrix}
					9.9620\cdot 10^{-4} &3.2355\cdot10^{-6}&9.0119\cdot10^{-6} \\
					3.2355\cdot10^{-5}&1.0112\cdot 10^{-3}&2.6781\cdot 10^{-6}\n \\
				9.0119\cdot10^{-6}&2.6781\cdot 10^{-6}&1.0090\cdot 10^{-3}
				\end{pmatrix}$\\\cline{4-7}
				&&& $\set{0,14}$&$0.2984$&$\begin{pmatrix}
				0.24946869 \\0.24983233
				\end{pmatrix}$&$\begin{pmatrix}
					 1.0076\cdot 10^{-3}&-1.3917\cdot 10^{-5} \\
					-1.3917\cdot 10^{-5}& 9.5561\cdot 10^{-4}
				\end{pmatrix}$\\\hline
				\multirow{2}{*}{a)}&\multirow{2}{*}{comp.\ wrapped}&\multirow{2}{*}{$15.3\pm 1.2$}&$\set{0,1,8} $&$0.7014$&$\begin{pmatrix}
					0.5000503 \\ 0.49934066 \\ 0.50020789
				\end{pmatrix}$&$\begin{pmatrix}
					99311\cdot10^{-4} &0&0 \\
					0&1.0204\cdot10^{-3}&0\n \\
					0&0&1.00854\cdot10^{-3}
				\end{pmatrix}$\\\cline{4-7}
				&&& $\set{0,14}$&$0.2986$&$\begin{pmatrix}
				0.24970119 \\0.2505222\n 
				\end{pmatrix}$&$\begin{pmatrix}
					1.0284\cdot 10^{-3}&0 \\
					0&1.0387\cdot 10^{-4}
				\end{pmatrix}$\\\hline
				\multirow{2}{*}{a)}&\multirow{2}{*}{von Mises}&\multirow{2}{*}{$15.9\pm 3.6$}&$\set{0,1,8}$&$0.7080$&$\begin{pmatrix}
				0.49985056\\ 0.50007666\\ 0.50037997
				\end{pmatrix}$&$\begin{pmatrix}
				25.6780\\ 26.1047\\  25.9094
				\end{pmatrix}$\\\cline{4-7}
				&&& $\set{0,14}$&$0.2920$&$\begin{pmatrix}
				0.24995573\\ 0.25006754
				\end{pmatrix}$&$\begin{pmatrix}
				25.4725\\ 25.0976
				\end{pmatrix}$\\\hline
				\multirow{2}{*}{b)}&\multirow{2}{*}{wrapped}&\multirow{2}{*}{$27.7\pm 6.7$}&$\set{0,1,8} $&$0.6992$&$\begin{pmatrix}
					0.49996806\\ 0.50002489 \\0.50002011
				\end{pmatrix}$&$\begin{pmatrix}
					1.0072\cdot 10^{-3} & 1.1067\cdot 10^{-5} & 1.1380\cdot10^{-6} \\
				1.1067\cdot 10^{-5} & 1.0151\cdot 10^{-3} &1.0376\cdot 10^{-5} \\
					1.1380\cdot10^{-6} &1.0375\cdot 10^{-5} & 9.8274\cdot 10^{-4}
				\end{pmatrix}$\\\cline{4-7}
				&&& $\set{0,14}$&$0.2017$&$\begin{pmatrix}
					0.25073647 \\ 0.2503585
				\end{pmatrix}$&$\begin{pmatrix}
					1.0154\cdot 10^{-3} & -7.2804 \cdot 10^{-6} \\
					-7.2804 \cdot 10^{-6} & 9.9143 \cdot 10^{-4}
				\end{pmatrix}$\\\cline{4-7}
				&&& $\set{5}$&$0.0991$&$0.29997122$&$9.6097\cdot 10^{-4}$\\\hline
				\multirow{3}{*}{b)}&\multirow{3}{*}{comp.\ wrapped}&\multirow{3}{*}{$28.8 \pm 6.1 $}&$\set{0,1,8} $&$0.7015$&$\begin{pmatrix}
					0.50011285\\ 0.50027382\\ 0.5001206
				\end{pmatrix}$&$\begin{pmatrix}
					1.0328\cdot 10^{-3} & 0 & 0 \\
					0 & 10017\cdot 10^{-3} & 0\\
					0 & 0& 10263 \cdot 10^{-3}
				\end{pmatrix}$\\\cline{4-7}
				&&& $\set{0,14}$&$0.1964$&$\begin{pmatrix}
				0.25049699 \\ 0.24964112
				\end{pmatrix}$&$\begin{pmatrix}
					9.6814 \cdot 10^{-4} & 0\\
					0 & 1.0540 \cdot 10^{-3}
				\end{pmatrix}$\\\cline{4-7}
				&&& $\set{5}$&$0.1021$&$0.30121135$&$9.4578\cdot 10^{-4}$\\\hline
				\multirow{3}{*}{b)}&\multirow{3}{*}{von Mises}&\multirow{3}{*}{$55.8\pm 11.9$}&$\set{0,1,8} $&$0.7054$&$\begin{pmatrix}
					0.49981447\\ 0.4998282\\  0.4993162 
				\end{pmatrix}$&$·\begin{pmatrix}
				25.3988\\ 24.2057\\ 26.1674
			\end{pmatrix}$\\\cline{4-7}
				&&& $\set{0,14}$&$0.1945$&$\begin{pmatrix}
					0.25106235\\ 0.2492706
				\end{pmatrix}$&$\begin{pmatrix}
				25.4221\\ 26.2056
			\end{pmatrix}$\\\cline{4-7}
				&&& $\set{5}$&$0.1001$&$0.29920271$&$24.9519$\\\hline
			\end{tabular}
		}\\
	\end{center}
	\caption{Parameters estimation of $f^i, i=1,2$ in example~\ref{ex:theor_ex_1}. 
		Top: $N=7500$, Bottom: $N=10000$, with the training parameters $\varepsilon_{c}=0.1, \gamma_{prox}=10^{-3}$ for all density functions and $\varepsilon_{KS}=3.4$ for the class of wrapped normal density function and  $\varepsilon_{KS}=2.7$ for the von Mises distribution function.}
	\label{tab:param_estimate_example_wrapped_normal_test_function}	
\end{table}
\subsection{B-spline functions}

In this section we will approximate the function $f^3:\T^9 \rightarrow \R $ defined as $f^3(x)=\tilde{f}^3(x)/||\tilde{f}_3||_{L^1},$ such that 
\begin{equation}
	\tilde{f}^3(x)=B_2(x_1)B_4(x_3)B_6(x_8)+B_2(x_2)B_4(x_5)B_6(x_6)+B_2(x_4)B_4(x_7)B_6(x_9),
\end{equation}
where $B_i, i\in set{2,4,6}$
To train the model $N_{tr}=10^4$ samples with density function $f^3$ has been drawn by the rejection sampling method. The ordered indices set of active variables are 
\begin{equation*}
	\mathcal{A}_{f^3}=\set{8,6,9,3,5,7,1,2,4}.
\end{equation*}
and the set of coupling indices 
\begin{equation}
	\mathcal{U}=\set{\set{1,3,8}, \set{2,5,6}, \set{4,7,9}}.
\end{equation}
Note that by the von Mises distribution there may exits more that one component (some times $2$ mixture components) with the same coupling variables. Similarly to the first two examples $N_{te}=10^5$ uniform samples has been used to estimate the model. The negative likelihood of the samples, the relative $L^1$ and $L_2$ errors are given in table~\ref{tab:log_likeli_rela_parame_b_spline}. We can also train the model by the usual EM-Algorithm–\ref{alg:alg_em_mm_orig} without any sparsity assumption on the density function. The computation time is \begin{equation}
	\mathcal{U}=\left(\set{1,3,8}, \set{2,5,6}, \set{4,7,9}\right)
\end{equation}

%%%%%%%%%%%%%%%%%%%%%%%%%%%%%%%%%%%%%%%%%
\begin{table}
	\centering
	\resizebox*{!}{2.cm}{
	\begin{tabular}{|l|l|l|l|l|}
		\hline
		\textbf{Method}&$\mathcal L_f(x^1,...,x^N)$&$\mathcal L_{\hat p}(x^1,...,x^N)$&$e_1(f,\hat{f})$&$e_2(f,\hat{f})$\\
		\hline
		wrapped&$6991.5\pm 117.6$&$6937.2\pm122.9$&$0.0808\pm0.0030$&$0.0802\pm 0.0036$\\
		\hline
		comp. wrapped &$6991.5\pm117.6 $&$6934.5\pm 121.8$&$0.0786\pm0.0022$&$0.0786\pm0.0031$\\
		\hline
		von Mises &$6931.7\pm 108.5$&$6669.1\pm 83.7$&$0.1600\pm0.0042$&$0.1515\pm 0.0056$\\
	%	Naiv EM-Algorithm&$6931.7$&&&\\
	\hline
	\end{tabular}
}
\end{table}

\begin{table}
	\centering
	\resizebox*{!}{2.cm}{
		\begin{tabular}{|l|l|l|}
			\hline
			\textbf{Method}&MSE&Time(s)\\
			\hline
			wrapped&$0.0257\pm 0.0023$&$97.3\pm 6.6$\\
			\hline
			comp. wrapped &$0.0247\pm 0.0019$&$102.4\pm 3.5$\\
			\hline
			von Mises &$0.0915\pm0.0067$&$152.6\pm 13.3$ \\
			\hline
			Naiv EM-Algorithm&&$10727\pm 560$\\
			\hline
		\end{tabular}
	}
\caption{Approximation of $f^3$ by  sparse mixture models of wrapped Gaussian and von Mises distribution from section~\ref{sec:SPAMM}. 
	Average value of the log likelihood function, relative $L_q$ errors, $q=1,2$  and Mean Square Error (MSE) for $N_{te}=10^5$ training sets. Top:negative Log-likelihood, relative errors. Bottom: MSE for all models and for naive EM-Algorithm.}
\label{tab:log_likeli_rela_parame_b_spline}
\end{table}

%%%%%%%%%%%%%%%%%%%%%%%%%%%%%%%%%%%%%%%%%
\subsection{California Housing Prices}
In the following we want to apply the California Housing Prices to our model. The data contain information from the  $1990$ California census. The goal is to predict the median house price $y$ with the help of  $8$ feature variables 
\begin{align*}
	1.&\text{ MedInc        median income in block}&5.& \text{ Population    block population}\\
	2.&\text{ HouseAge      median house age in block}	&6.&\text{ AveOccup      average house occupancy}\\
	3.&\text{ AveRooms      average number of rooms}&7. &\text{ Latitude      house block latitude}\\
	4.&\text{ AveBedrms     average number of bedrooms}& 8.&\text{ Longitude     house block longitude}
\end{align*}
%Therefore, we assume that the target variable $y\in \T$ can be approximated by a convex combination of linear regression models. 
%The mixture of $K$-components of a regression model is given by 
%\begin{equation}
%	y=\begin{cases}
%		\beta_1x_{u_1}+\gamma_1, \text{ with probability } \alpha_1 ,\\\
%		\beta_2x_{u_2}+\gamma_2,  \text{ with probability } \alpha_2,\\
%		\quad\vdots\qquad \quad\vdots\hspace{3.5cm}\vdots\\
%		\beta_Kx_{u_K}+\gamma_K,  \text{ with probability } \alpha_K,\\
%	\end{cases}
%\end{equation}
%where for any mixture component $k\in [K]$ the parameters $\left(\beta_{k}\right)^{\intercal}\in \R^{\abs{u_k}} , \gamma \in \R$ and $\alpha=(\alpha_1, \cdots,\alpha_K)\in \Delta_K.$

The total number of samples $N_{CH}$ is $20640$. The input variable will be denoted by $x$ and the target variable by $y$. A data preprocessing step has cleaned the data and using using min-max scaler to rescale the data to fit $\T$. We assume that the input features $x^n, n\in N_{CH}$ are the samples of a $8$-dimensional random variable $X\in \T^d$ and the target variable $y^n, n\in N_{CH}$ a sample from $Y\in \T.$

To be able to train the data with Algorithm~\ref{alg:density_func_approx} we assume that the target density function can be written as convex combination of linear model, i.e
$$y=\varphi(x)=\sum_{k=1}^K \alpha_k \left(M_kx_{u_k}+N_k\right),$$
where for all $k\in [K]$  $M_k\in\R^{1\times \abs{u_k}}$ and $N_k\in \R.$
 Recall  the regression model can be rewritten as conditional expectation $\varphi(x)=E(y|x),$ where the conditional expectation is defined as
\begin{equation}
	E(y| x) =\int_{\T} yp_{Y| X}(y\mid x,\theta)dy
\end{equation}
and $p_{Y| X}(y| x,\theta)$ defined the conditional density function of $Y| X.$ Furthermore, we assume that the joint density function is sparse with the form~\ref{eq:sparse_density_function}. 
If $p_{Y| X}$ denotes the conditional density function of a wrapped Gaussian distribution, we know by Theorem~\ref{lem:conditional_distr_marginal_dist} that
\begin{align}
		E(y| x,l)  &=\sum_{k}\alpha_k(x)\left(\mu_{u^c_k} + \Sigma_{u^c_ku_k}\left(\Sigma_{u_ku_k}\right)^{-1}\left(x_{u_k}+l_{u_k}-\mu_{u_k}\right)\right)
\end{align}
where the weights $\alpha_k$ are defined as
\begin{equation}
	\alpha_k(x) =\frac{\tilde{\alpha}_kp_{X_{u_k}}(x_{u_k}| \theta_k)}{\sum_{k=1}^K\tilde{\alpha}_kp_{X_{u_k}}(x_{u_k}| \theta_k)}.
\end{equation}
Note that the density function 
\begin{equation}
p_{X_{\mathcal{A}}}(x|\theta)=\sum_{k=1}^K\tilde{\alpha}_kp_{X_{u_k}}(x_{u_k}| \theta_k)
\end{equation}
is the marginal density function with respect to $X_{\mathcal{A}}$  with parameters $\left(\alpha_k,\mu_{u_k},\Sigma_{u_ku_k}\right)_{k\in[K]}$ and 
\begin{equation}
f(z)=p(z|\theta)=p(x_{\mathcal{A}},y|\theta)=\sum_{k=1}^{\tilde{K}}\alpha_kp(y,x_{u_k}|\theta_k)
\end{equation}
 the joint density function $Z_{\mathcal{A}}=\left(X_{\mathcal{A}},Y\right).$ Similarly to the previous section $\mathcal{A}$ denotes the set of active variable of $f$.  

Obviously the set of active variable are the variables $X_i$ such that $X_i$ and $Y$ are correlated and $X_i$ is non uniformly distributed. Therefore the index set of active variables are
\begin{equation}
	\mathcal{A}_{CH} =\set{1,3,6,9},
\end{equation}
 the collection of coupling variables are 
\begin{equation}
	\mathcal{U}=\left(\set{1,3,9}\right)
\end{equation}
and the MSE of the approximated model is given in Table~\ref{tab:MSE_Calfifornia_Housing_Data}. Comparing this model with some other regression model show a better approximation result that Linear Regression, Lasso Regression, Ridge Regression,  Decision Tree Regression and Random Forest Regression.
\begin{table}
	\centering
	\resizebox*{!}{3.cm}{
		\begin{tabular}{|l|l|}
			\hline
			\textbf{Method}&MSE\\
			\hline
			Wrapped&$\boldsymbol{0.1803}$\\
			\hline
			%comp. wrapped &\\
			%\hline
			%von Mises & \\
			%\hline
			Linear Regression&$0.4478$\\
			\hline
			Lasso Regression&$0.7193$\\
			\hline
			Ridge Regression &$0.6056$\\
			\hline
			Decision Tree Regression&$0.5912$\\
			\hline 
			Random Forest Regression&$0.4478$\\
			\hline
		\end{tabular}
	}
\caption{Mean Square Error of California Housing Prices by SPAMM and other regression Models}
\label{tab:MSE_Calfifornia_Housing_Data}
\end{table}
\section{Conclusion}
This paper introduces an efficient algorithm that can accurately estimate the parameters of a sparse mixture model of wrapped Gaussian and von Mises distribution based on the input samples. Assuming that each component of the multivariate density function depends only on a certain number of interacting variables, which is also unknown, we have iteratively determined the set $\mathcal{U}$ using statistical tests and the model parameters using Expectation-Maximization. Incorporating this sparsity assumption speeds up the learning procedure in the case where the dimension $d$ is very large and even provides a better approximation accuracy. Indeed this yields a better approximation accuracy and is more efficient than the usual expectation maximization for the mixtures of wrapped Gaussian, the B-spline function, and the California housing data, as shown in the numerical results. However, the approximation relies on the choice of hyperparameters, which, when not chosen appropriately, leads to underfitting or overfitting of the mixture model. 
%\appendix
%\bibliographystyle{plain}
%\bibliography{mm_anova_lit}
%\nocite{*}

\section{Statistical Methods} \label{sec:tests}
\subsection{Correlation Test} \label{sec:ct}
To test whether the features are uncorrelated, we have to verify if their correlation coefficients are zero. Recall that the correlation coefficient of two random variables $X_i$ and $X_j$ is define as
\begin{equation*}
	\text{Cor}(X_i, X_j) = \displaystyle\frac{\text{Cov}(X_i, X_j)}{\sigma^2_{X_i}\sigma^2_{X_j}}, 
\end{equation*}
where $\text{Cov}(X_i, X_j)$ represents the covariance of $X_i$ and $X_j$ and  $\sigma^2_{X_i},\sigma^2_{X_j}$ their variance.
If some weighted samples of the random variables are provided, the coefficient are approximated by using their corresponding unbiased weighted samples covariance matrix 
\begin{equation*}
	\Sigma := \frac{1}{\sum_{n=1}^N w_n -1}\sum_{n=1}^N w_n(x^n - \mu)^{\intercal}(x^n - \mu),
\end{equation*}
where 
\begin{equation*}
	\mu := \frac{1}{\sum_{n=1}^N w_n}\sum_{n=1}^N w_n x^n
\end{equation*}
is the weighted samples mean. 

\subsection{Kolmogorov-Smirnov Test} \label{sec:ks}
Recall that the Kolmogorov-Smirnov test \cite{monahan_2001} is a statistical test often used to test if some given samples $\mathcal{X}$ fit a distribution whose cumulative density function $F$ is a priori known. The samples fit $F$  if the Kolmogorov-Smirnov (KS) distance is smaller than a threshold, i.e
\begin{equation}
	D_N(x^n,w_n)=D_N(F,F_N)=\sqrt{\frac{\left(\sum_{n=1}^Nw_n\right)^2}{\sum_{n=1}^Nw_n^2}}\norm{F-F_N}_{\infty} \leq \varepsilon_{KS},
\end{equation}
where $F_N$ represents the empirical cumulative density function 
\begin{equation}
	F_N = \frac{1}{\sum_{n=1}^{N}w_n}\sum_{n=1}^N w_n \mathds{1}_{[x^n,1]}
\end{equation}
of the weighted samples $\mathcal{X}$. Let 
\begin{equation*}
	s_n = \frac{\sum_{m=1}^n w_m}{\sum_{n=1}^N w_n},
\end{equation*} 
if we assume that the samples are ordered increasing then the KS distance becomes
\begin{equation}
	D_N(F,F_N)=\sqrt{\frac{\left(\sum_{n=1}^Nw_n\right)^2}{\sum_{n=1}^Nw_n^2}}\max\limits_{n=1,\ldots, N}\max\set{s_n-x^n, x^n-s_{n-1}},
\end{equation}
and will be denoted by $D_N\left(\left(x^n,w_n\right)_{n\in [N]}\right).$
\subsection{EM Algorithm} \label{sec:em}
%--------------------------------------------------------------------------------
%\textcolor{red}{- with prox part from ANOVA paper
%	- try to incorporate modification (some summands are fixed)}

In the following we want to approximate the parameters of the samples distribution under the assumption that the parameters of some fixed mixture components are already known. As already mentioned above, this can be done by the usual Expectation maximization algorithm (EM) (see algorithm~\ref{alg:alg_em_mm_orig}). To ensure sparse model in the mixing weights, the proximal expectation maximization (Prox-EM) algorithm~\ref{alg:prox_EM} can be used instead. 
\begin{algorithm}[!ht]
	\caption{EM Algorithm for Mixture Models}\label{alg:alg_em_mm_orig}
	\begin{algorithmic}[1]
		\State \textbf{Input:} $(x^1,...,x^N)\in\T^{N,d}$, $(w_1,\ldots,w_N) \in \mathbb R^N,$  
		initial parameters $(\alpha^{(0)}_k,\theta^{(0)}_{u_k})_{k \in [K]}.$ 
		\State \textbf{Output:} Optimal parameters $(\alpha, \Theta)$
		\For {$r=0,1,...$}
		\State \textbf{E-Step:} For $k=1,...,K$ and $n=1,\ldots,N$ compute 
		$$
		\beta_{n,k}^{(r)}
		=\frac{\alpha_k^{(r)}p(x_{u_k}^n|\theta_{u_k}^{(r)})}{\sum_{j=1}^{K}\alpha_j^{(r)}p(x_{u_j}^n|\theta_{u_j}^{(r)})}
		$$
		\State \textbf{M-Step:} For $k=1,...,K_1$ compute
		\begin{align}
			\alpha_k^{(r+1)}&=\frac{1}{\sum_{m=1}^N w_m}  \sum_{n=1}^N w_n \beta_{n,k}^{(r)},\\
			\theta_k^{(r+1)}&=\argmax_{\theta_k}\Big\{\sum_{i=1}^{N} w_n \beta_{n,k}^{(r)} \log(p (x^n_{u_k}|\theta_{u_k}^{(r)}))\Big\}.
		\end{align}
		\EndFor
	\end{algorithmic}
\end{algorithm}
In the Prox-EM algorithm, the goal is to minimize the penalized functional
\begin{equation*}
	\mathcal{L_{\gamma}}(\alpha, \Theta \mid \mathcal{X}) = \mathcal{L}(\alpha, \Theta \mid \mathcal{X}) + \gamma \norm{\alpha}_0 + \iota_{\Delta_K}(\alpha),
\end{equation*}
where $\gamma >0$ represents the learning rate.

If we assume that the samples $\mathcal{X}$ can be fitted by a $K$-components mixture model with parameters $(\alpha_k, \theta_k)_{k \in [K]}$ and  for a given $K_1,$ $\abs{K_1}<K,$ the mixtures parameters $(\alpha_k, \theta_k)_{k \in K_1}$ are known, we have to modify the EM-Algorithm to only approximate the parameters $(\alpha_k,\theta_k), k \in K_1^c=[K]\setminus K_1.$ Note that this cannot be done directly, by simply fixing the a priori known parameters in the expectation and maximization step. Therefore we will split the mixture components into two groups $G_1$ and $G_2,$ where $G_1$ contains components with parameters $(\alpha_k, \theta_k)_{k\in K_1}$ and $G_2$ those with parameters $(\alpha_k, \theta_k)_{k\in K_1^c}.$ Weighting the target distribution samples $\mathcal{X}$ with the posterior probability 
\begin{equation*}
	\beta_{n,G_2} = \frac{\sum_{k \in K_1^c} \alpha_k p(x_{u_k}\mid \theta_{u_k})}{\sum_{j \in K} \alpha_j p(x_{u_j}\mid \theta_{u_j})},
\end{equation*}
that they belong to $G_2$ times the ground samples weights, one can use the above described EM-algorithm to find the parameters of the mixture model who fits the best the weighted samples $\set{(x_n, w_n\beta_{n,G_2})_{n \in [N]}}.$ The output parameters are those of the mixture model from Group $G_2,$ except to the mixing weights. The mixture components weights has to be rescaled, such that their sum is equal $\alpha_{G_2}=1-\sum_{k \in K_1} \alpha_k,$ by multiplying each term with $\alpha_{G_2}$. To get a better accuracy of the EM-algorithm, we have chosen as initialization parameters the elements of \ref{}, where the mixing weights has been also rescaled by dividing them with $\alpha_{G_2}.$ For explicit details of the EM-algorithm of the wrapped and the von-Mises distribution see \cite{hertrich2021sparse}.

\begin{algorithm}[!ht]
	\caption{Proximal Expectation Maximization (Prox-EM)}\label{alg:prox_EM}
	\begin{algorithmic}[1]
		\State \textbf{Input:} $(x^1,...,x^N)\in\T^{N,d}$, $(w_1,\ldots,w_N) \in \mathbb R^N,$  initial parameters $(\alpha^{(0)}_k,\theta^{(0)}_{u_k})_{k \in [K]},$ learning rate $\gamma >0.$ 
		\State \textbf{Output:} Optimal parameters $(\alpha, \Theta)$
		\For {$r=0,1,...$}
		\State Set $K_c =\abs{\Theta}$
		\State \textbf{E-Step:} For $k=1,...,K_c$ and $n=1,\ldots,N$ compute 
		$$
		\beta_{n,k}^{(r)}
		=\frac{\alpha_k^{(r)}p(x_{u_k}^n|\theta_{u_k}^{(r)})}{\sum_{j=1}^{K_c}\alpha_j^{(r)}p(x_{u_j}^n|\theta_{u_j}^{(r)})}
		$$
		\State \textbf{M-Step:} For $k=1,...,K_c$ compute
		\begin{align}
			\alpha_k^{(r+1)}&=\frac{1}{\sum_{m=1}^N w_m} \sum_{n=1}^N w_n \beta_{n,k}^{(r)},\\
			\theta_k^{(r+1)}&=\argmax_{\theta_k}\Big\{\sum_{i=1}^{N} w_n \beta_{n,k}^{(r)} \log(p (x^n_{u_k}|\theta_{u_k}^{(r)}))\Big\}.
		\end{align}
		\State Order the mixing weights $(\alpha_k)_{k \in K_c}$ and compute the number $K_0,$ where 
		\begin{equation}
			K_0 = \argmin_{m =0,\cdots, K_c-1} \frac{1}{2\gamma}\left(\frac{\left(\sum_{k=1}^m \alpha_k\right)^2}{K_c-m} + \sum_{k=1}^{m}\alpha_k^2-m \right)
		\end{equation} 
		\State Set $J = [K_c]\setminus [K_0]$ and $J^c=[K_0]$ such that
		\begin{align}
			\alpha &= (\hat{\alpha}_k)_{k\in J}, \text{ where } \hat{\alpha}_k = \alpha_k + \frac{1}{\abs{J}} \sum_{k \in J^c} \alpha_k,\\
			\Theta &= \set{\theta_{u_k}}_{k \in J}
		\end{align}
		\EndFor
	\end{algorithmic}
\end{algorithm}
%--------------------------------------------------------------------------------

\subsection{Bayesian Information Criteria (BIC)} \label{sec:BIC}
The Bayesian Information Criterion (BIC) \cite{mclachlan2004finite, BrochadoVitorino2005} is a statistical method introduced by Schwarz in 1978 for model selection. Given a finite number of models, the BIC is based on the Likelihood the model which 
Given a maximal number of components $K_{\max}$, we will iteratively train the model with $K_b=1, \ldots, K_{max}.$ The optimal number of components $K_{opt}$ is the one with the optimal Bayesian Information Criterion (BIC)
\begin{equation}\label{eq: opt_k_BIC}
	\text{BIC}(k) := -2\mathcal{L}(\alpha_{tr}, \theta_{tr} \mid \mathcal{X}_i) + 3k-1,
\end{equation}
where
\begin{equation*}
	\mathcal{L}(\alpha, \theta \mid \mathcal{X}_i) = \sum_{n=1}^N \tilde{w}_{n,k}^i \log \left(\sum_{k=1}^{K_b}\alpha_{b, k} p\left(x_i^n\mid \theta_{b,k}\right)\right)
\end{equation*}
denotes the likelihood of the samples with respect to the trained parameters $\alpha_{b}, \theta_{b}$ corresponding to $K_b$ components. Let
\begin{equation}\label{eq:param_estim_univ_cond_dist}
	\alpha\left(u_k^{r(i)}, i^\prime\right)\in \Delta_{K\left(u_k^{r(i)}, i^\prime\right)}, \text{ and } \theta\left(u_k^{r(i)}, i^\prime\right) 
\end{equation}  
be the estimated parameters associated to the optimal number of mixture component $K\left(u_k^{r(i)}, i^\prime\right)$ given by \eqref{eq: opt_k_BIC}. 
%Note that the BIC is unbiased but however have a large variance. The MCCV-LL in contrast to BIC is biased but  

\section{Proofs}
Before showing Lemma~\ref{lem:conditional_distr_marginal_dist} we will formulate a proposition about the inverse and the determinant of a block matrix
\begin{proposition}\cite{zhang2011matrix}\label{prop:det_inverse_block_matrix}
	Let $\Sigma  \in \R_+^{d \times d}$ be a positive definite block matrix with the form
	\begin{equation}
		\Sigma = \begin{pmatrix}
			A & B \\ C & D
		\end{pmatrix}
	\end{equation}
	where $A \in \R_+^{n \times n}, B \in \R_+^{n \times m}, C  \in \R_+^{m \times n}, D \in \R_+^{m \times m},$ such that $m+n=d.$ Then the following assertions hold
	\begin{enumerate}[i.]
		\item\label{prop:det_blockmatrix} $\det(\Sigma)= \det(A)\cdot\det(D-CA^{-1}B),$ if $A$ is invertible,
		\item The matrix $\Sigma$ is invertible with inverse
		\begin{equation*}
			\Sigma^{-1} = \begin{pmatrix}
				\tilde{A} & \tilde{B} \\ \tilde{C} & \tilde{D}
			\end{pmatrix},
		\end{equation*}
		where 
		\begin{equation}\label{eq:inv_block_matrix1}
			\begin{cases}
				\tilde{A} &\coloneqq A^{-1}+A^{-1}B\left(D-CA^{-1}B\right)^{-1}CA^{-1},\\
				\tilde{B} &\coloneqq -A^{-1}B\left(D-CA^{-1}B\right)^{-1},\\
				{C} &\coloneqq -\left(D-CA^{-1}B\right)^{-1}CA^{-1},\\
				\tilde{D} &\coloneqq -\left(D-CA^{-1}B\right)^{-1}.
			\end{cases}
		\end{equation}\label{eq:inv_block_matrix2}
		By applying permutation, the block matrices of the inverse matrix $\Sigma^{-1}$ become
		\begin{equation}
			\begin{cases}
				\tilde{A} &\coloneqq \left(A-BD^{-1}C\right)^{-1},\\
				\tilde{B} &\coloneqq -\left(A-BD^{-1}C\right)^{-1}BD^{-1},\\
				{C} &\coloneqq -D^{-1}C\left(A-BD^{-1}C\right)^{-1},\\
				\tilde{D} &\coloneqq D^{-1}+D^{-1}C\left(A-BD^{-1}C\right)^{-1}BD^{-1}.
			\end{cases}
		\end{equation}
	\end{enumerate}
\end{proposition}
\begin{proof}
	Let $n\coloneqq \abs{u}.$ We will only prove the assumption for the marginal of $f$ with respect to $X_u.$ The marginal pdf with respect to $X_u$ is given by
	\begin{align*}
		f_{X_u}(x_u) &= \int_{\T^{n_c}} f(x) dx_{u^c} = \int_{\T^{n_c}} \wN(x \mid \mu, \Sigma) dx_{u^c},\\
		&= \sum_{l \in \Z^d}\int_{\T^{n_c}} \mathcal{N}(x+l \mid \mu, \Sigma) dx_{u^c}, \\
		&= \sum_{l \in \Z^d}\int_{\T^{n_c}} \frac{1}{\sqrt{(2\pi)^d \abs{\Sigma} }} \exp\left(-1/2(x+l-\mu\right)^{\intercal} \Sigma^{-1}\left(x+l-\mu) \right)dx_{u^c}.
	\end{align*}
	For all $u \subset [d]$ the matrix $\Sigma$ can be decomposed into a $2\times2$ block matrix
	\begin{equation*}
		\Sigma = \begin{pmatrix}
			\Sigma_{uu}& \Sigma_{uu^c} \\
			\Sigma_{u^cu} & \Sigma_{u^cu^c}
		\end{pmatrix}.
	\end{equation*}
	Proposition~\ref{prop:det_inverse_block_matrix}~\eqref{prop:det_blockmatrix} implies that 
	\begin{equation}\label{eq:marg_det_block_matr}
		\det(\Sigma) = \det(\Sigma_{uu})\cdot \det(\Sigma_{u^cu^c}-\Sigma_{u^cu}\left(\Sigma_{uu}\right)^{-1}\Sigma_{uu^c})	
	\end{equation}
	Furthermore by proposition~\ref{prop:det_inverse_block_matrix}~\eqref{eq:inv_block_matrix1} the inverse covariance matrix yields
	\begin{equation*}
		\Sigma^{-1} = \begin{pmatrix}
			\tilde{\Sigma}_{uu} & \tilde{\Sigma}_{uu^c}\\
			\tilde{\Sigma}_{u^cu} & \tilde{\Sigma}_{u^cu^c} 
		\end{pmatrix},
	\end{equation*}
	where 
	\begin{align*}
		\tilde{\Sigma}_{uu} &=\left(\Sigma_{uu}\right)^{-1}+\left(\Sigma_{uu}\right)^{-1}\Sigma_{uu^c} \left(\Sigma_{u^cu^c}-\Sigma_{u^cu}\left(\Sigma_{uu}\right)^{-1}\Sigma_{uu^c}\right)^{-1}\Sigma_{u^cu}\left(\Sigma_{uu}\right)^{-1},\\ \tilde{\Sigma}_{uu^c}&=-\left(\Sigma_{uu}\right)^{-1}\Sigma_{uu^c}\left(\Sigma_{u^cu^c}-\Sigma_{u^cu}\left(\Sigma_{uu}\right)^{-1}\Sigma_{uu^c}\right)^{-1},\\
		\tilde{\Sigma}_{u^cu}&=-\left(\Sigma_{u^cu^c}-\Sigma_{u^cu}\left(\Sigma_{uu}\right)^{-1}\Sigma_{uu^c}\right)^{-1}\Sigma_{u^cu}\left(\Sigma_{uu}\right)^{-1},\\
		\tilde{\Sigma}_{u^cu^c}&=\left(\Sigma_{u^cu^c}-\Sigma_{u^cu}\left(\Sigma_{uu}\right)^{-1}\Sigma_{uu^c}\right)^{-1}.
	\end{align*}
	Define $y^l = (x^l-\mu)=(x+l-\mu) \in \T^d$ and $\mu_{l_u}= \Sigma_{u^cu}\left(\Sigma_{uu}\right)^{-1}y^l_u,$ since each $d$-dimensional vector $y$ can be decomposed into $y=(y_u, y_{u^c})^{\intercal}.$ Thus
	\begin{align*}
		\left(x+l -\mu\right)\Sigma^{-1}\left(x+l -\mu\right) &= 
		(y^l_u)^{\intercal} (\Sigma^{-1})_{uu}y^l_u + (y^l_u)^{\intercal} (\Sigma^{-1})_{uu^c}y^l_{u^c}\\
		&\qquad  +  (y^l_{u^c})^{\intercal} (\Sigma^{-1})_{u^cu}(y^l_{u}) + (y^l_{u^c})^{\intercal} (\Sigma^{-1})_{u^cu^c}y^l_{u^c}\\
		&= \text{\rom{1}+\rom{2}+\rom{3}+\rom{4}}.
	\end{align*}
	Replacing the block matrices of the inverse covariance matrix yields
	\begin{align*}
		\text{\rom{1}} &= (y^l_u)^{\intercal} (\Sigma_{uu})^{-1}y^l_u + (y^l_u)^{\intercal} (\Sigma_{uu})^{-1}\Sigma_{uu^c}\left(\Sigma_{u^cu^c}-\Sigma_{u^cu}\left(\Sigma_{uu}\right)^{-1}\Sigma_{uu^c}\right)^{-1}\Sigma_{u^cu}\left(\Sigma_{uu}\right)^{-1}y^l_u,\\
		&=  (y^l_u)^{\intercal} (\Sigma_{uu})^{-1}y^l_u + (y^l_u)^{\intercal} (\Sigma_{uu})^{-1}\Sigma_{uu^c}\left(\Sigma_{u^cu^c}-\Sigma_{u^cu}\left(\Sigma_{uu}\right)^{-1}\Sigma_{uu^c}\right)^{-1}\Sigma_{u^cu}\left(\Sigma_{uu}\right)^{-1}y^l_u,\\
		&=(y^l_u)^{\intercal} (\Sigma_{uu})^{-1}y^l_u + \mu_{l_u}^{\intercal} \tilde{\Sigma}_{u^cu^c} \mu_{l_u}= \text{\rom{1}(a)+\rom{1}(b)}.
	\end{align*}
	and also
	\begin{align*}
		\text{\rom{2}} &= -(y^l_u)^{\intercal}(\Sigma_{uu})^{-1}\Sigma_{uu^c}\left(\Sigma_{u^cu^c}-\Sigma_{u^cu}\left(\Sigma_{uu}\right)^{-1}\Sigma_{uu^c}\right)^{-1}y^l_{u^c}=-\mu_{l_u}^{\intercal}\tilde{\Sigma}_{u^cu^c}y^l_{u^c},\\
		\text{\rom{3}} &=-(y^l_{u^c})^{\intercal} \left(\Sigma_{u^cu^c}\Sigma_{u^cu}\left(\Sigma_{uu}\right)^{-1}\Sigma_{uu^c}\right)^{-1}\Sigma_{u^cu}\left(\Sigma_{uu}\right)^{-1}(y^l_{u})= -(y^l_{u^c})^{\intercal} \tilde{\Sigma}_{u^cu^c}\mu_{l_u},\\
		\text{\rom{4}} &= (y^l_{u^c})^{\intercal} \left(\Sigma_{u^cu^c}-\Sigma_{u^cu}\left(\Sigma_{uu}\right)^{-1}\Sigma_{uu^c}\right)^{-1}y^l_{u^c} = (y^l_{u^c})^{\intercal}\tilde{\Sigma}_{u^cu^c}y^l_{u^c}
	\end{align*}
	Summing up \rom{1}(b), \rom{2}, \rom{3} and \rom{4} together yields
	\begin{equation*}
		\text{\rom{1}(b)+\rom{2}+\rom{3}+\rom{4}}= \left(y^l_{u^c} - \mu_{l_u}\right)^{\intercal}\tilde{\Sigma}_{u^cu^c}\left(y^l_{u^c} - \mu_{l_u}\right)=\left(y^l_{u^c} - \mu_{l_u}\right)^{\intercal}\left(\Sigma^{-1}\right)_{u^cu^c}\left(y^l_{u^c} - \mu_{l_u}\right).
	\end{equation*}
	This implies that
	\begin{equation}\label{eq:marg_wrapp_decomp_exp_term}
		\left(x+l -\mu\right)\Sigma^{-1}\left(x+l -\mu\right) = (y^l_u)^{\intercal} (\Sigma_{uu})^{-1}y^l_u + \left(y^l_{u^c} - \mu_{l_u}\right)^{\intercal}\left(\Sigma^{-1}\right)_{u^cu^c}\left(y^l_{u^c} - \mu_{l_u}\right).
	\end{equation} 
	Equations~\eqref{eq:marg_wrapp_decomp_exp_term} and \eqref{eq:marg_det_block_matr} imply that for a fixed $l \in \Z^d,$ the shifted normal density function can be decomposed as
	\begin{equation}\label{eq:decomposition_wrapped_gaussian_pdf}
		\mathcal{N}\left(x+l \mid \mu, \Sigma\right) = \mathcal{N}\left((x+l)_u \mid \mu_u, \Sigma_{uu}\right) \cdot \mathcal{N}\left((x+l)_{u^c} \mid \mu_{u^c}-\mu_{l_u}, \tilde{\Sigma}_{u^cu^c}\right)
	\end{equation}
	Inserting \eqref{eq:decomposition_wrapped_gaussian_pdf} into the definition of the marginal density function $f_{X_u}$ gives us
	\begin{align*}
		f_{X_u} &= \sum_{l=(l_u, l_{u^c})^{\intercal} \in \Z^d}\mathcal{N}\left((x+l)_u \mid \mu_u, \Sigma_{uu}\right) \cdot  \int_{\T^{n_c}}\mathcal{N}\left((x+l)_{u^c} \mid \mu_{u^c}-\mu_{l_u}, \tilde{\Sigma}_{u^cu^c}\right)dx_{u^c}\\
		&=\sum_{l_u\in \Z^{n}}\mathcal{N}\left((x+l)_u \mid \mu_u, \Sigma_{uu}\right) \cdot \sum_{l_{u^c}\in \Z^{n_c}} \int_{\T^{n_c}}\mathcal{N}\left((x+l)_{u^c} \mid \mu_{u^c}-\mu_{l_u}, \tilde{\Sigma}_{u^cu^c}\right)dx_{u^c}\\
		&= \sum_{l_u\in \Z^{n}}\mathcal{N}(x_u+l\mid \mu_u, \Sigma_{uu})\cdot \sum_{t_1\in \Z}\int_{I_1}\ldots\sum_{t_{n_c}\in \Z}\int_{I_{n^c}}\mathcal{N}(z\mid \mu_{u^c}-\mu_{l_u}, \tilde{\Sigma}_{u^cu^c})dz_{u^c}, \\
		&=\sum_{l_u\in \Z^{n}}\mathcal{N}(x_u+l\mid \mu_u, \Sigma_{uu})\cdot \int_{\R}\ldots\int_{\R}\mathcal{N}(z\mid \mu_{u^c}-\mu_{l_u}, \tilde{\Sigma}_{u^cu^c})dz_{u^c}, \\
		&=\sum_{l_u\in \Z^{n}}\mathcal{N}(x_u+l\mid \mu_u, \Sigma_{uu})\cdot \underbrace{\int_{\R^{n_c}}\mathcal{N}(z\mid \mu_{u^c}- \mu_{l_u}, \tilde{\Sigma}_{u^cu^c})dz_{u^c}}_{ =1}, \\
	\end{align*}
	where $I_j = [t_j, t_j+1[,$ for all $j=1,\ldots, n^c,$ and $l_{u^c}=(t_1, \ldots, t_{n^c})$. This implies the claim. We can similarly show that the assumption of the marginal density function $f_{X_{u^c}}$ with respect to $X_{u^c}$ also holds, by using the second definition of the inverse block matrix from proposition~\ref{prop:det_inverse_block_matrix}~\ref{eq:inv_block_matrix2} for the inverse covariance matrix $\Sigma^{-1}$ where
	\begin{align*}
		\tilde{\Sigma}_{uu}  &=\left(\Sigma_{uu}-\Sigma_{uu^c}\left(\Sigma_{u^cu^c}\right)^{-1}\Sigma_{u^cu}\right)^{-1}\\
		\tilde{\Sigma}_{uu^c}&=-\left(\Sigma_{uu}-\Sigma_{uu^c}\left(\Sigma_{u^cu^c}\right)^{-1}\Sigma_{u^cu}\right)^{-1}\Sigma_{uu^c}\left(\Sigma_{u^cu^c}\right)^{-1},\\
		\tilde{\Sigma}_{u^cu}&=-\left(\Sigma_{u^cu^c}\right)^{-1}\Sigma_{u^cu}\left(\Sigma_{uu}-\Sigma_{uu^c}\left(\Sigma_{u^cu^c}\right)^{-1}\Sigma_{u^cu}\right)^{-1},\\
		\tilde{\Sigma}_{u^cu^c}&=\left(\Sigma_{u^cu^c}\right)^{-1}+\left(\Sigma_{u^cu^c}\right)^{-1}\Sigma_{u^cu}\left(\Sigma_{uu}-\Sigma_{uu^c}\left(\Sigma_{u^cu^c}\right)^{-1}\Sigma_{u^cu}\right)^{-1}\Sigma_{uu^c}\left(\Sigma_{u^cu^c}\right)^{-1},
	\end{align*}
	and by applying proposition~\ref{prop:det_inverse_block_matrix}~\eqref{prop:det_blockmatrix} to the determinant of the inverse covariance matrix we then get
	\begin{align*}
		\det(\Sigma) &= \sfract{1}{\det(\Sigma^{-1})}=\sfract{1}{\left(\det(\tilde{\Sigma}_{uu})\cdot\det\left(\tilde{\Sigma}_{u^cu^c}-\tilde{\Sigma}_{u^cu}\left(\tilde{\Sigma}_{uu}\right)^{-1}\tilde{\Sigma}_{uu^c} \right)\right)}, \\
		&=\sfract{1}{\det\left(\left(\Sigma_{uu}-\Sigma_{uu^c}\left(\Sigma_{u^cu^c}\right)^{-1}\Sigma_{u^cu}\right)^{-1}\right)\cdot \det\left(\left(\Sigma_{u^cu^c}\right)^{-1}\right)},\\
		&=\det\left(\Sigma_{uu}-\Sigma_{uu^c}\left(\Sigma_{u^cu^c}\right)^{-1}\Sigma_{u^cu}\right)\cdot \det\left(\Sigma_{u^cu^c}\right),
	\end{align*}
	since the determinant of a matrix is also equal to the inverse of the determinant of its inverse matrix. 
	
	Furthermore by the Bayes theorem the conditional density function holds for all $x\in \T^d$ and $u \subset [d]$ 
	\begin{equation*}
		f_{X_{u^c} \mid (X, L)_{u}} = \sum_{l_{u^c} \in \Z^{n_c}}f_{(X,L)_{u^c} \mid (X, L)_{u}} =\sum_{l_{u^c} \in \Z^{n_c}}\frac{f_{(X, L)}}{f_{(X, L)_{u}}}.
	\end{equation*}
	Hence together with \eqref{eq:decomposition_wrapped_gaussian_pdf} we obtain
	\begin{align}
		f(x_{u^c}\mid (x,l)_u)&=\frac{\sum_{l_{u^c} \in \Z^{n_c} }\mathcal{N}(x+l \mid \mu, \Sigma)}{\mathcal{N}((x+l)_u\mid \mu_u, \Sigma_{uu})},\\ 
		&=\sum_{l_{u^c} \in \Z^{n_c} }\frac{\mathcal{N}((x+l)_{u} \mid \mu_u, \Sigma_{uu})\mathcal{N}((x+l)_{u^c} \mid \mu_{u^c}-\mu_{l_u}, \tilde{\Sigma}_{u^cu^c})}{\mathcal{N}(\tilde{x}+\tilde{l}\mid \theta_u)} , \\
		&= \sum_{l_{u^c}\in \Z^{n_c}}\frac{\mathcal{N}((x+l)_u \mid \mu_u, \Sigma_{uu}) \mathcal{N}((x+l)_{u^c} \mid \mu_{u^c}-\mu_{l_u}, \tilde{\Sigma}_{u^cu^c})}{\sum_{l_u\in \Z^{n}}\mathcal{N}((x+l)_u \mid \mu_u, \Sigma_{uu})},\\
		&= \sum_{l_{u^c}\in \Z^{n_c}}\mathcal{N}((x+l)_{u^c} \mid \mu_{u^c}-\mu_{l_u}, \tilde{\Sigma}_{u^cu^c}).
	\end{align}
\end{proof}
\end{document}